\documentclass{article}
\usepackage{preamble} 
\usepackage{xcolor}

\title{Bootstrapping Classical and Quantum Query Lower Bounds for Finding Tarski Fixed Points}

\author[1,2]{Tongyang Li\thanks{tongyangli@pku.edu.cn}}
\author[3]{Weiran Ma\thanks{2400013063@stu.pku.edu.cn}}
\author[1,2]{Ziyi Yang\thanks{yangziyi@stu.pku.edu.cn}}
\author[3]{Xingyu Zhao\thanks{2300012903@stu.pku.edu.cn}}
\affil[1]{Center on Frontiers of Computing Studies, Peking University, Beijing 100871, China}
\affil[2]{School of Computer Science, Peking University, Beijing 100871, China}
\affil[3]{School of Electronics Engineering and Computer Science, Peking University, Beijing 100871, China}
\date{}

\begin{document}

\thispagestyle{empty}

\maketitle

\begin{abstract}
The Tarski fixed-point problem asks for a fixed point of a monotone map on a $k$-dimensional grid lattice of width $n$ and has applications in game theory, economics, and verification. Substantial gaps remain between the known upper and lower bounds on its quantum query complexity.

We prove two hardness bootstrapping theorems for quantum query lower bounds. Starting from a hard instance family whose instances have pair-wise incomparable unique fixed points, and a lower bound established by the nonnegative adversary method, we bootstrap this lower bound by (i) a factor of $\Omega(t)$ when the dimension is multiplied by $t$ and the width remains unchanged, and (ii) a factor of $\Omega(\log m)$ when the width is multiplied by $m$ and the dimension is doubled. Applying these theorems, we obtain three consequences: (i) an $\Omega(k\log^2 n)$ lower bound for $k\geq 4$, improving previously best-known classical and quantum lower bounds; (ii) for every fixed constant $\alpha>0$, there exists a constant $\delta(\alpha)>0$ such that $Q_\epsilon(\Tarski(n,k))=\Omega(kn^{\delta(\alpha)})$ whenever $k>n^\alpha$, and in particular, an $\Omega(kn^{0.0072})$ lower bound for $k>n^{2/203}$, the first lower bound with simultaneous linear dependence on the dimension and polynomial dependence on the width; and (iii) a superpolynomial lower bound of $(k/\log k)^{\Omega(\log k)}$ for exponentially wide grids, where $n=2^{\Theta(k)}$. The last result rules out a route to polynomial-time algorithms for simple stochastic games and ARRIVAL through generic polynomial-query black-box Tarski solvers.
\end{abstract}

\tableofcontents

\pagebreak

\section{Introduction}

The Knaster--Tarski fixed-point theorem \cite{tarski1955ALF} guarantees the existence of a fixed point for every monotone self-map of a complete lattice. Beyond its central role in order theory, the theorem has numerous applications in computer science and economics. In programming languages, least and greatest fixed points underpin recursive program semantics, monotone data-flow analysis, abstract interpretation, and fixed-point logics used in verification and model checking \cite{kildall1973unified,cousot1977abstract,kozen1982results}. In economics and game theory, fixed-point formulations characterize equilibria in supermodular games and clearing states in financial networks \cite{topkis1979equilibrium,eisenberg2001systemic}.

We consider the problem of finding a Tarski fixed point on the grid lattice $[n]^k$, where $[n]=\{0,\ldots,n-1\}$ and the partial order is defined by ${x}\preceq{y}$ if and only if $x_i\leq y_i$ for every $i\in\{1,\ldots,k\}$. Given oracle access to a monotone function $f\colon[n]^k\rightarrow[n]^k$, the goal is to find ${x}$ satisfying $f({x})={x}$; we denote this problem by $\Tarski(n,k)$. In complexity theory, this problem provides a common framework for several game theory problems. Finding a pure Nash equilibrium in a two-player supermodular game, with strategy space $[n]^k$ for each player and best-response access, is computationally equivalent to $\Tarski(n,k)$ \cite{etessami2020tarski}. Computing the exact values of an $N$-vertex simple stochastic game, introduced by Condon \cite{condon1992complexity}, reduces to $\Tarski(2^{\operatorname{poly}(N)},N)$ \cite{etessami2020tarski}. Similarly, an instance of $\ARRIVAL$ \cite{dohrau2017arrival} with $N$ switching vertices reduces to $\Tarski(2^N+1,N)$ \cite{gartner2021subexponential}. The decision versions of the latter two problems belong to $\NP\cap\coNP$, but whether they belong to $\P$ remains open. These reductions offer a potential route to polynomial-time solutions through a generic Tarski algorithm running in time polynomial in $k$ and $\log n$.

Research on the query complexity of $\Tarski(n,k)$ has pursued both algorithms and lower bounds. On the upper-bound side, iterating $f$ from the bottom or top of the grid uses $O(kn)$ queries \cite{etessami2020tarski}, while recursive coordinate-wise binary search gives an $O((\log n)^k)$-query algorithm for fixed $k$ \cite{dang2025computations}. Fearnley, P\'alv\"olgyi, and Savani \cite{fearnley2022faster} obtained an $O(\log^2 n)$-query algorithm in three dimensions; combining their decomposition theorem with the one-, two-, and three-dimensional algorithms yields an $O((\log n)^{\lceil 2k/3\rceil})$ upper bound. Chen and Li \cite{chen2022improved} introduced the auxiliary problem $\Tarski^*$ and a decomposition theorem for it, avoiding the repeated binary-search overhead incurred by decomposing $\Tarski$ directly and reducing the exponent to $\lceil(k+1)/2\rceil$. Chen, Li, and Yannakakis \cite{chen2026mystery} subsequently developed a safe-function framework and a geometric candidate-elimination argument to solve $\Tarski^*(n,3)$ in $O(\log n)$ queries. This gives an optimal $O(\log^2 n)$ upper bound in four dimensions and, through decomposition, an $O((\log n)^{1+\lceil(k-1)/3\rceil})$ upper bound in general fixed dimension. Very recently, Chen, Gao, Li, Rubinstein, and Yannakakis \cite{chen2026quasifpt} obtained the first quasi-FPT query upper bound, $O(5^k(\log n)^{\lceil\log_2 k\rceil})$ for $k\geq 3$. Their recursive algorithm approximately halves the dimension at each level, reducing the exponent of $\log n$ from linear to logarithmic in $k$.

Until recently, substantially less was understood about lower bounds. Etessami, Papadimitriou, Rubinstein, and Yannakakis \cite{etessami2020tarski} constructed two-dimensional herringbone functions around a hidden monotone spine, proving an $\Omega(\log^2 n)$ classical randomized query lower bound. Br\^anzei, Phillips, and Recker \cite{branzei2025randomized} constructed a family of hard instances on the Boolean hypercube and proved an $\Omega(k)$ lower bound. The same authors \cite{branzei2025tarski} generalized the herringbone construction to higher-dimensional grids, obtaining an $\Omega(k\log^2 n/\log k)$ randomized query lower bound. In the quantum query model, Phillips \cite{phillips2026quantum} established a composition theorem for nonnegative adversary bounds with generalized-search inner functions. By embedding nested ordered search into two-dimensional herringbone functions, an $\Omega(\log^2 n)$ quantum query lower bound is established. Li, Ma, Yang, and Zhao \cite{li2026quantum} constructed project-to-the-spine hard instances, which project grid points onto a hidden monotone spine and then move them toward its midpoint. They developed the Tree-Filtration Adversary Method to analyze this hierarchical construction and establish an $\Omega(k\log n)$ quantum query lower bound. Chen, Li, and Yannakakis \cite{chen2026mystery} highlighted the possibility of a generic $O(f(k)\log^2 n)$ query upper bound. 

\subsection{Our Contributions}

We prove two hardness bootstrapping theorems for the quantum query complexity of finding Tarski fixed points. Starting from a nonnegative adversary lower bound for $\Tarski(n,k)$ on a hard instance family satisfying suitable conditions---specifically, every function has a unique fixed point and all possible fixed points form an antichain (see the formal definition in \eqref{eq:antichain-definition})---we construct hard instances on grids of different dimensions and widths with amplified lower bounds. For any fixed $0\leq\epsilon<1/2$, these hardness bootstrapping theorems yield the following consequences:
\begin{enumerate}
\item \textbf{Improved logarithmic exponents in constant dimensions.}
We obtain an $\Omega(\log^3 n)$ lower bound for $k\geq 8$, an $\Omega(\log^4 n)$ lower bound for $k\geq 32$, and so on. This rules out the possibility of a generic $O(f(k)\log^2 n)$ upper bound. Combined with the upper bounds in \cite{chen2026quasifpt}, this determines the classical and quantum query complexity up to constant factors in dimensions $k=8,9,10$, where $D(\Tarski(n, k)) = R_\epsilon(\Tarski(n, k)) = Q_\epsilon(\Tarski(n,k))=\Theta(\log^3 n)$.

\item \textbf{A dimension-linear, log width-square lower bound.}
For every $k\geq 4$, we prove $Q_\epsilon(\Tarski(n,k))=\Omega(k\log^2 n)$. This strengthens the previously known quantum lower bound $\Omega(\log^2 n+k\log n)$ \cite{phillips2026quantum, li2026quantum} and classical randomized lower bound $\Omega(k+k\log^2 n/\log k)$ \cite{branzei2025randomized, branzei2025tarski}, capturing simultaneous linear dependence on $k$ and quadratic dependence on $\log n$.

\item \textbf{Polynomial dependence on width in high dimensions.}
For every fixed constant $\alpha>0$, there exists a constant $\delta(\alpha)>0$ such that $Q_\epsilon(\Tarski(n,k))=\Omega(kn^{\delta(\alpha)})$ whenever $k>n^\alpha$. In particular, taking $\alpha=2/203$ gives $Q_\epsilon(\Tarski(n,k))=\Omega(kn^{0.0072})$.

\item \textbf{Superpolynomial lower bounds at exponential width.}
For $n=2^{\Theta(k)}$, we prove $Q_\epsilon(\Tarski(n,k))\geq(k/\log k)^{\Omega(\log k)}$. This gives a superpolynomial query lower bound in the parameter regime arising from reductions of $\ARRIVAL$ and $\SSG$ to Tarski. Consequently, a generic black-box Tarski solver cannot provide a route to proving $\ARRIVAL,\SSG\in\BQP$, not even to mention $\P$.
\end{enumerate}

Specifically, our first theorem arranges $t$ instances from such a family into a chain, multiplying the nonnegative adversary lower bound by a factor of $\Omega(t)$ while increasing the dimension from $k$ to $tk$:

\begin{restatable}[Dimensional Hardness Bootstrapping Theorem]{theorem}{dimensionalbootstrapping}
\label{thm:dimensional-hardness-bootstrapping}
Let $\mathcal{F}$ be a family of hard instances for $\Tarski(n,k)$ satisfying the following conditions:
\begin{enumerate}
\item Every $f\in\mathcal{F}$ has a unique fixed point, denoted by $f^*$.
\item The set of all possible fixed points,
$A:=\{f^*\mid f\in\mathcal{F}\}$, is an antichain.
\item There exists a nonnegative adversary matrix $\Gamma$ for $\Tarski(n,k)|_{\mathcal{F}}$ that yields a spectral lower bound of $\lambda/\delta$.
\end{enumerate}
Then, for every positive integer $t$,
\begin{align}
\SA(\Tarski(n,tk))
\geq
\frac{t\lambda}{C\delta},
\end{align}
where $C=5+2\sqrt{2}$.
\end{restatable}

For a hard instance family satisfying the same conditions, increasing the width from $n$ to $nm$ and the dimension from $k$ to $2k$ multiplies its nonnegative adversary lower bound by a factor of $\SA(\OS_m) = \Omega(\log m)$. Specifically, we first transform each $\Tarski(n,k)$ instance into a $\Tarski^*$ instance that hides an arbitrary sign $s\in\{-1,0,1\}$ at its unique solution. We then compose the problem of recovering this hidden sign with one-dimensional ordered search and embed the resulting abstract problem into a grid of dimension $2k$ and width $nm$.

\begin{restatable}[Logarithmic Hardness Bootstrapping Theorem]{theorem}{logarithmicbootstrapping}
\label{thm:logarithm-bootstrapping}
Let $\mathcal{F}$ be a hard instance family for $\Tarski(n,k)$ in which every function has a unique fixed point and all possible fixed points lie on an antichain $A$. Suppose that there exists a nonnegative adversary matrix $\Gamma_0$ for $\Tarski(n,k)|_{\mathcal{F}}$ that yields a spectral lower bound of $\lambda/\delta$. Then, for every integer $m\geq 2$, there exists a hard instance family $\mathcal{F}'$ for $\Tarski(nm,2k)$ in which every function has a unique fixed point, together with a corresponding nonnegative adversary matrix $\Gamma$ that yields a spectral lower bound of $\SA(\OS_m)\lambda/(2\delta)$.
\end{restatable}

A general hard instance family need not satisfy the antichain condition required by these transformations. However, a dimension-doubling transformation can place all possible fixed points of any unique-fixed-point family on an antichain. We denote it by $\mathcal{B}$ and show that it preserves nonnegative adversary lower bounds up to a constant factor, allowing its use in the quantum query model.

\begin{restatable}{theorem}{antichainlifting}
\label{thm:antichain-lifting}
Let $\mathcal{F}$ be a hard instance family for $\Tarski(n,k)$ in which every function has a unique fixed point. Suppose that there exists a nonnegative adversary matrix $\Gamma_0$ for $\Tarski(n,k)|_{\mathcal{F}}$ that yields a spectral lower bound of $\lambda/\delta$. Then $\mathcal{B}(\mathcal{F})$ is a hard instance family for $\Tarski(n,2k)$ in which every function has a unique fixed point and all possible fixed points lie on an antichain. Moreover, there exists a nonnegative adversary matrix $\Gamma$ for $\Tarski(n,2k)|_{\mathcal{B}(\mathcal{F})}$ that yields a spectral lower bound of $\lambda/(2\delta)$.
\end{restatable}

\subsection{Concurrent Work}

During the preparation of this manuscript, Liu and Wang \cite{liu2026dimension} and Chen, Gao, Li, Rubinstein, and Yannakakis \cite{chen2026quasifpt} also obtained the same classical antichain-lifting transformation and similar logarithmic bootstrapping theorems. Compared with these concurrent works, our results offer the following advances:
\begin{enumerate}
\item In addition to the two results above, our dimensional hardness bootstrapping theorem is entirely new, which multiplies the lower bound by $\Omega(t)$ when the dimensionality is multiplies by $t$.

\item Our lower bounds are proved under the quantum query setting; they directly apply to classical deterministic and bounded-error randomized query complexities.

\item Our embedding of the abstract composite problem is more compact: gaining an additional $\log m$ factor requires dimension expansion from $k$ to $4k$, rather than $4k+1$ and $4k+2$ in the two concurrent works, respectively. Consequently, we obtain a tight $\Omega(\log^3 n)$ quantum query lower bound already in dimension $k=8$, whereas the corresponding classical bounds become tight starting at $k=9$ and $k=10$ in the two concurrent works, respectively.
\end{enumerate}

\subsection{Paper Organization}

In Section~\ref{sec:preliminaries}, we formally define the problems $\Tarski$, $\Tarski^*$, and $\OS$, review the quantum query model and the nonnegative adversary method, and present two constructive techniques for establishing nonnegative adversary lower bounds. In Section~\ref{sec:dimensional-bootstrapping}, we describe a construction that links $t$ hard instances into a chain and construct a nonnegative adversary matrix to prove Theorem~\ref{thm:dimensional-hardness-bootstrapping}. In Subsection~\ref{subsec:conjugate-transformation}, we prove Theorem~\ref{thm:antichain-lifting}. In Section~\ref{sec:logarithmic-bootstrapping}, we introduce the sign-hiding construction, compose the resulting problem with $\OS_m$, and embed the composite problem into $[nm]^{2k}$ to prove Theorem~\ref{thm:logarithm-bootstrapping}. In Section~\ref{sec:consequences}, we establish the consequences outlined above. Finally, in Section~\ref{sec:conclusion}, we summarize our results and conclude with directions for further improvement.

\section{Preliminaries} \label{sec:preliminaries}

\subsection{\texorpdfstring
  {The {\mdseries\textsc{Tarski}} and {\mdseries\textsc{Tarski}}$^*$ Problems}
  {The Tarski and Tarski* Problems}}

We begin by reviewing some basic properties of lattices and formally introducing the \textsc{Tarski} and \textsc{Tarski}$^*$ problems.

A partially ordered set $(\mathcal{L}, \preceq)$ is called a \emph{complete lattice} if every subset $S \subseteq \mathcal{L}$ has an infimum $\inf S$ and a supremum $\sup S$ in $\mathcal{L}$. A subset $A \subseteq \mathcal{L}$ is called an \emph{antichain} if any two distinct elements of $A$ are incomparable; equivalently,
\begin{align}
\forall x, y\in A, \quad (x\preceq y \vee y\preceq x)\rightarrow x = y. \label{eq:antichain-definition}
\end{align}

A function $f\colon \mathcal{L}\rightarrow \mathcal{L}$ is called \emph{monotone} if, for all $x, y\in \mathcal{L}$, $x\preceq y$ implies $f(x)\preceq f(y)$. Tarski's fixed-point theorem~\cite{tarski1955ALF} states that every monotone function on a complete lattice has a fixed point. If $f$ have unique fixed point, we denote the unique fixed point by $f^*$. This motivates the study of fixed-point computation on specific lattices.

Let $n$ and $k$ be positive integers, and let $[n] = \{0, 1, \ldots, n - 1\}$. Equip $[n]^k$ with the coordinatewise partial order $\preceq$ defined by
\begin{align}
\forall x, y\in [n]^k, \quad x\preceq y \Leftrightarrow \forall i \in \{1, \ldots, k\},\quad x_i \leq y_i.
\end{align}
The resulting complete lattice $([n]^k, \preceq)$ is called a $k$-dimensional \emph{grid lattice} of width $n$. The problem of finding a fixed point on a grid lattice, denoted by $\textsc{Tarski}(n, k)$, is defined as follows.

\begin{definition}[$\textsc{Tarski}(n, k)$]
Given oracle access to a monotone function $f$ on the grid lattice $([n]^k, \preceq)$, find any point $x\in [n]^k$ such that $f(x) = x$.
\end{definition}

When $k=1$, $\Tarski(n,k)$ can be solved by binary search. For a positive integer $n$, the ordered search problem is formally defined as follows.

\begin{definition}[$\OS_n$]\label{def:ordered-search}
For each $i\in\{1,\ldots,n\}$, the ordered search instance with hidden position $i$ is the function $f_i\colon\{1,\ldots,n\}\rightarrow\{-1,0,1\}$ defined by
\begin{align}
f_i(j)=
\begin{cases}
1, & j<i,\\
0, & j=i,\\
-1, & j>i.
\end{cases}
\end{align}
Given oracle access to some $f_i$, find the hidden position $i$.
\end{definition}

The $\textsc{Tarski}^*$ problem is a standard subproblem that arises in algorithms for the Tarski fixed-point problem (see~\cite{chen2022improved}). When performing binary search along a chosen dimension, solving this subproblem suffices to determine which direction to continue the search. The problem is defined as follows.

\begin{definition}[$\textsc{Tarski}^*(n, k)$]
Given oracle access to a monotone function $f\colon [n]^k\rightarrow [n]^k$ on $([n]^k, \preceq)$ and a monotone sign function $\sigma\colon [n]^k\rightarrow \{-1, 0, 1\}$, find any point $x\in [n]^k$ satisfying either $x\preceq f(x)\wedge \sigma(x)\geq 0$ or $x\succeq f(x)\wedge \sigma(x)\leq 0$. Here, monotonicity of $\sigma$ means that $\sigma(x)\leq \sigma(y)$ whenever $x\preceq y$.
\end{definition}

Note that previous works on upper bounds (for example \cite{chen2022improved}) additionally require that every coordinate of $f(x)-x$ belong to $\{-1, 0, 1\}$ for every $x\in [n]^k$. Our lower-bound argument in this paper does not require this restriction; we only assume that $f$ is monotone.

\subsection{Quantum Query Model} \label{subsec:quantum-query-model}

We work in the standard quantum query model. Let $\Sigma$ and $\Pi$ be finite alphabets, and let $I$ be an index set. Our goal is to evaluate a partial function $\Phi\colon S\rightarrow \Pi$, where $S\subseteq \Sigma^I$. The input $x\in S$ is not given explicitly; instead, it is accessed through an oracle $O_x$, defined for all $i\in I$ and $z\in \Sigma$ by
\begin{align}
O_x\ket{i, z} = \ket{i, z\oplus x_i},
\end{align}
where $\oplus$ denotes an abelian group operation on $\Sigma$.

A $T$-query quantum algorithm starts in a state $\ket{\psi_0}$ independent of the input $x$ and applies the following operations in sequence:
\begin{align}
U_0\rightarrow O_x \rightarrow U_1 \rightarrow O_x \rightarrow \cdots \rightarrow U_{T - 1} \rightarrow O_x \rightarrow U_T \rightarrow \text{Measurement.}
\end{align}
Here, $U_0, \ldots, U_T$ are unitaries independent of $x$. The final step is a computational-basis measurement of a designated output register whose basis states are labeled by elements of $\Pi$.

The algorithm computes $\Phi$ with error at most $\epsilon$ if, for every $x\in S$, the measurement outcome equals $\Phi(x)$ with probability at least $1-\epsilon$. The bounded-error quantum query complexity of $\Phi$, denoted by $Q_\epsilon(\Phi)$, is the minimum nonnegative integer $T$ for which there exists a $T$-query quantum algorithm that computes $\Phi$ with error at most $\epsilon$.

\subsection{Nonnegative Adversary Method}

The nonnegative adversary method is a standard technique for proving quantum query lower bounds. Its spectral formulation was introduced by Barnum, Saks, and Szegedy~\cite{barnum2003quantum} and was subsequently shown to be equivalent to other nonnegative-weight adversary methods by \cite{spalek2006all}. We follow the formulation of the spectral adversary bound in \cite[Theorem~3.1]{spalek2006all}. Throughout this subsection, we use the computational model and notation introduced in Section \ref{subsec:quantum-query-model}. We begin by defining two types of matrices that we will need.

\begin{definition}[Nonnegative adversary matrix]
A matrix $\Gamma\in \R_{\geq 0}^{|S|\times |S|}$ is a \emph{nonnegative adversary matrix} for $\Phi$ if it is symmetric and satisfies
\begin{align}
\forall x, y\in S, \quad
\Phi(x) = \Phi(y) \rightarrow \Gamma[x, y] = 0. \label{eq:nonnegative-adversary-condition}
\end{align} 
\end{definition}

\begin{definition}[Distinguisher matrix]
For each $i\in I$, the $i$-th \emph{distinguisher matrix} is the matrix $\Delta_i\in \{0, 1\}^{|S|\times |S|}$ defined by
\begin{align}
\Delta_i[x, y] =
\begin{cases}
1, & \text{if $x_i \ne y_i$}, \\
0, & \text{otherwise.}
\end{cases}
\end{align}
\end{definition}

Define
\begin{align}
\SA(\Phi)
:=
\max_{\Gamma\ne 0}
\frac{\norm{\Gamma}}
{\max_{i\in I}\norm{\Gamma\circ \Delta_i}},
\end{align}
where the maximum is taken over all nonzero nonnegative adversary matrices for $\Phi$, $\circ$ denotes the Hadamard (entrywise) product, and $\norm{\cdot}$ denotes the spectral norm. If $\Phi$ is constant, we set $\operatorname{SA}(\Phi) = 0$. For $\delta>0$, we say that an adversary matrix $\Gamma$ for $\Phi$ \emph{yields a spectral lower bound of $\lambda / \delta$} if
\begin{align}
\norm{\Gamma} \geq \lambda
\quad \text{and} \quad
\max_{i\in I} \norm{\Gamma \circ \Delta_i} \leq \delta.
\end{align}
Note that this implies $\SA(\Phi)\geq \lambda / \delta$.

The following theorem gives a lower bound on the bounded-error quantum query complexity of $\Phi$.

\begin{theorem}[Nonnegative adversary method, see~\cite{barnum2003quantum}]
For every $0\leq \epsilon < 1/2$,
\begin{align}
Q_\epsilon(\Phi)
\geq
\frac{1 - 2\sqrt{\epsilon(1 - \epsilon)}}{2}
\SA(\Phi).
\end{align}
\end{theorem}

The following lower bound for ordered search is proved using the nonnegative adversary method in \cite[Appendix~B]{phillips2026quantum}:

\begin{lemma} \label{lem:ordered-search-lowerbound}
For every integer $m\geq 2$, we have $\SA(\OS_m) \geq (\ln m - 2)/\pi$.
\end{lemma}

\subsection{Techniques for Adversary Matrix Construction}

Our hardness bootstrapping results are obtained by constructing new adversary matrices from existing ones and their associated lower bounds. This process consists of two stages: constructing new abstract problems from known problems, and embedding these abstract problems into grid lattices. These two stages rely, respectively, on the constructive techniques presented in Sections~\ref{subsubsec:composition} and~\ref{subsubsec:reduction}.

\subsubsection{Generalized Search Composition} \label{subsubsec:composition}

In this paper, we construct hard instances using function composition. Phillips~\cite{phillips2026quantum} provides a method for proving quantum query lower bounds for composed functions. This method requires the inner functions to be generalized search functions, whose definition we restate below using the notation from Section~\ref{subsec:quantum-query-model}.

\begin{definition}[Generalized search function, see~{\cite[Definition~4]{phillips2026quantum}}] \label{def:generalized-search}
Suppose that a function $\Phi\colon S \rightarrow \Pi$ satisfies the following conditions:
\begin{enumerate}
\item There exists a positive integer $m\in \N$ such that $|\Phi^{-1}(\pi)| = m$ for every $\pi\in \Pi$. Thus, the inputs can be uniquely labeled by pairs $(\pi, j)$, where $j\in [m]$ and $\Phi((\pi, j)) = \pi$.
\item This labeling can be chosen so that, for each fixed query index $i\in I$ and any two inputs $(\sigma_1, j_1), (\sigma_2, j_2)\in S$, whether or not $(\sigma_1, j_1)_i = (\sigma_2, j_2)_i$ depends only on $j_1$, $j_2$, and whether or not $\sigma_1 = \sigma_2$, rather than on the particular values of $\sigma_1$ and $\sigma_2$.
\end{enumerate}
Then $\Phi$ is called a \emph{generalized search function} with $m$ \emph{variants}.
\end{definition}

The composition of search problems is defined as follows (see \cite[Definition~3]{phillips2026quantum}).

\begin{definition} \label{def:composite-function}
Let $\Sigma,\Pi,\Xi$ be finite alphabets, and let $J$ be a finite index set. Let $\Psi\colon T\rightarrow\Xi$ be a function with $T\subseteq\Pi^J$. For each $j\in J$, let $I_j$ be a finite index set and let $\Phi_j\colon S_j\rightarrow\Pi$ be a function with $S_j\subseteq\Sigma^{I_j}$. Define the composite index set as the disjoint union
\begin{align}
I := \bigsqcup_{j\in J} I_j
   = \{(j,i)\mid j\in J,\ i\in I_j\}.
\end{align}
For each $x\in\Sigma^I$, write $x=(x^{(j)})_{j\in J}$, where $x^{(j)}\in\Sigma^{I_j}$ is defined by $x^{(j)}_i=x_{(j,i)}$. The \emph{composition} $\Psi\circ(\Phi_j)_{j\in J}\colon S_{\mathrm{comp}}\rightarrow\Xi$ is defined by
\begin{align}
\bigl(\Psi\circ(\Phi_j)_{j\in J}\bigr)(x)
:=
\Psi\bigl((\Phi_j(x^{(j)}))_{j\in J}\bigr),
\end{align}
where the domain of composite function is defined by
\begin{align}
S_{\mathrm{comp}}
:=
\bigcup_{y\in T}
\prod_{j\in J}\Phi_j^{-1}(y_j)
\subseteq\Sigma^I.
\end{align}
\end{definition}

Phillips~\cite{phillips2026quantum} proves the following nonnegative adversary bound for function composition. The corresponding adversary matrix construction is given in \cite[Definition~7]{phillips2026quantum}.

\begin{theorem}[Generalized search composition, see~{\cite[Theorem~2]{phillips2026quantum}}] \label{thm:generalized-search-composition}
Let $\Psi\circ(\Phi_j)_{j\in J}$ be a composition as in Definition \ref{def:composite-function}, where $J$ is a nonempty finite index set. Suppose that $\Phi_j$ is a generalized search function (Definition \ref{def:generalized-search}) for every $j\in J$. Then
\begin{align}
\SA\bigl(\Psi\circ(\Phi_j)_{j\in J}\bigr)
\geq
\SA(\Psi)\cdot\min_{j\in J}\SA(\Phi_j).
\end{align}
\end{theorem}

\subsubsection{Adversary Bounds via Query Reductions} \label{subsubsec:reduction}

After constructing an abstract composite problem, we must embed it into a grid lattice and construct an adversary matrix for the resulting family of hard instances. The following lemma shows that the adversary matrix and its associated lower bound can be readily transferred through such an embedding.

\begin{lemma}[Adversary Reduction Lemma] \label{lem:adversary-reduction}
Let $\Sigma_1,\Pi_1,\Sigma_2,\Pi_2$ be finite alphabets, and let $I_1,I_2$ be finite index sets. Consider two functions $\Phi_1\colon S_1\rightarrow\Pi_1$ and $\Phi_2\colon S_2\rightarrow\Pi_2$, where $S_1\subseteq\Sigma_1^{I_1}$ and $S_2\subseteq\Sigma_2^{I_2}$. Let $\mathcal{F}_1\subseteq S_1$ be a family of hard instances for $\Phi_1$, and suppose that a nonnegative adversary matrix $\Gamma$ for $\Phi_1|_{\mathcal{F}_1}$ yields a spectral lower bound of $\lambda / \delta$.

Suppose that there exists an injective map $\iota\colon\mathcal{F}_1\rightarrow S_2$ satisfying the following conditions:
\begin{enumerate}
\item The map $\iota$ preserves equality of outputs: for every $f_1,f_2\in\mathcal{F}_1$,
\begin{align}
\Phi_2(\iota(f_1))=\Phi_2(\iota(f_2))
\quad\Longrightarrow\quad
\Phi_1(f_1)=\Phi_1(f_2). \label{eq:iota-preserve-equality}
\end{align}

\item Each query position in the new problem depends, for distinguishing inputs, on a constant number of query positions in the original problem. Specifically, there exist a positive integer $c$ and maps $q_1,\ldots,q_c\colon I_2\rightarrow I_1$ such that, for every $z\in I_2$,
\begin{align}
\forall f_1,f_2\in\mathcal{F}_1,\quad
\iota(f_1)(z)\ne\iota(f_2)(z)
\quad\Rightarrow\quad
\exists j\in\{1,\ldots,c\},\;
f_1(q_j(z))\ne f_2(q_j(z)). \label{eq:difference-constant-dependence}
\end{align}
\end{enumerate}

Then, letting $\mathcal{F}_2=\iota(\mathcal{F}_1)$, there exists a nonnegative adversary matrix $\Gamma'$ for $\Phi_2|_{\mathcal{F}_2}$ that yields a spectral lower bound of $\lambda / (c\delta)$.
\end{lemma}

\begin{proof}
Since $\iota$ is injective, we have $|\mathcal{F}_2|=|\mathcal{F}_1|$. We obtain $\Gamma'$ simply by relabeling the rows and columns of $\Gamma$. Specifically, define
\begin{align}
\Gamma'[\iota(f_1),\iota(f_2)]
=
\Gamma[f_1,f_2].
\end{align}

First, $\Gamma'$ is indeed a nonnegative adversary matrix for $\Phi_2|_{\mathcal{F}_2}$. Relabeling preserves the nonnegativity and symmetry of $\Gamma$, while condition~\eqref{eq:iota-preserve-equality} ensures that $\Gamma'$ also satisfies condition~\eqref{eq:nonnegative-adversary-condition}. Moreover, relabeling preserves the spectral norm, so
$\norm{\Gamma'}=\norm{\Gamma}\geq\lambda$.

The upper bound on the denominator follows from condition~\eqref{eq:difference-constant-dependence}. Let $\Delta_i^{(1)}$ and $\Delta_z^{(2)}$ denote the distinguisher matrices for $\Phi_1|_{\mathcal{F}_1}$ and $\Phi_2|_{\mathcal{F}_2}$, respectively. Identifying the rows and columns indexed by $\mathcal{F}_2$ with those indexed by $\mathcal{F}_1$ via $\iota$, this condition gives, for every $z\in I_2$,
\begin{align}
\Delta_z^{(2)}
\leq
\sum_{j=1}^c \Delta_{q_j(z)}^{(1)},
\end{align}
where $\leq$ denotes entrywise comparison. By the monotonicity of the spectral norm under entrywise domination of nonnegative matrices and the triangle inequality, we obtain
\begin{align}
\norm{\Gamma'\circ\Delta_z^{(2)}}
\leq
\sum_{j=1}^c
\norm{\Gamma\circ\Delta_{q_j(z)}^{(1)}}
\leq c\delta.
\end{align}

Therefore, $\Gamma'$ yields a spectral lower bound of $\lambda / (c\delta)$ for $\Phi_2|_{\mathcal{F}_2}$.
\end{proof}

\section{Dimensional Hardness Bootstrapping Theorem} \label{sec:dimensional-bootstrapping}

In this section, we prove the following dimensional hardness bootstrapping theorem: Given a hard instance family for a low-dimensional Tarski fixed-point problem, increasing the dimension by $t$ while keeping the grid width unchanged amplifies the lower bound by a factor of $\Omega(t)$. 

\dimensionalbootstrapping*

In Section~\ref{subsec:chain-hard-instance}, we present a chain-like hard instance construction that can be embedded into $[n]^{tk}$. In Section~\ref{subsec:dimensional-spectral-analysis}, we construct the corresponding adversary matrix and carry out its spectral analysis. Applying the theorem requires that all possible fixed points of the hard instance family form an antichain. General hard instance families need not satisfy this condition; in particular, the herringbone family used in \cite{phillips2026quantum} does not. In Section~\ref{subsec:conjugate-transformation}, we present a transformation that lifts the possible fixed points of an arbitrary hard instance family with unique fixed points onto an antichain.

\subsection{Constructing a Chain of Hard Instance} \label{subsec:chain-hard-instance}

Let $\mathcal{F}$ be a family of hard instances for $\Tarski(n,k)$ such that every function has a unique fixed point and all possible fixed points lie in an antichain $A$. We arrange $t$ hard instances $f_1,\ldots,f_t\in\mathcal{F}$ in a chain so that, intuitively, any algorithm that finds the fixed point must solve the constituent problems one after another.

Specifically, consider a query $(x_1,\ldots,x_t)$ for which the first $i$ blocks are fixed points of their respective functions, namely $f_1(x_1)=x_1,\ldots,f_i(x_i)=x_i$, while the next block is not, namely $f_{i+1}(x_{i+1})\ne x_{i+1}$. We return $x_1,\ldots,x_i$ in the first $i$ blocks and $f_{i+1}(x_{i+1})$ in the $(i+1)$-st block. In the remaining blocks, we reveal no information about $f_{i+2},\ldots,f_t$; instead, we return padding values $\bot=\mathbf{0}$ or $\top=(n-1)\mathbf{1}$ that can be computed from the available information, as illustrated in Figure~\ref{fig:chain-visualization}.

\begin{figure}[htbp]
    \centering
    \includegraphics[width=\linewidth]{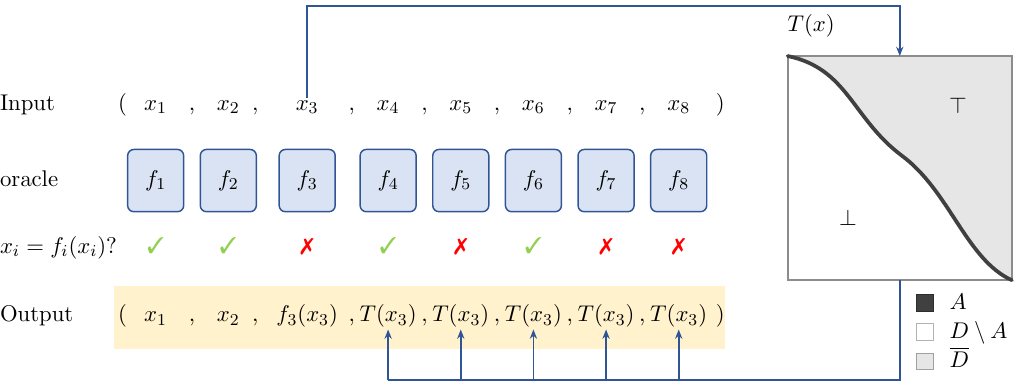}
    \caption{Visualization of the chain construction.}
    \label{fig:chain-visualization}
\end{figure}

Formally, let $D:=\{y\in[n]^k : \exists x\in A,\ y\preceq x\}$ be the downward closure of the antichain $A$ containing all possible fixed points. Define the padding function $T\colon[n]^k\rightarrow[n]^k$ by
\begin{align}
T(x) :=
\begin{cases}
\bot, & x\in D, \\
\top, & x\notin D.
\end{cases}
\label{eq:t-padding}
\end{align}

It is immediate from the definition of $D$ that $T$ is monotone.

For a tuple of functions $\boldsymbol{f}=(f_1,\ldots,f_t)$ and a tuple of query points $\boldsymbol{x}=(x_1,\ldots,x_t)$, define the first index at which the queried point is not a fixed point of the corresponding function:
\begin{align}
\sigma(\boldsymbol{f},\boldsymbol{x}) :=
\begin{cases}
t+1, & \forall i\in\{1,\ldots,t\},\ f_i(x_i)=x_i, \\
\min\{i\in\{1,\ldots,t\} : f_i(x_i)\ne x_i\},
& \text{otherwise.}
\end{cases}
\label{eq:sigma-first-unfixed}
\end{align}

The following chain constructor transforms this tuple of functions into a hard instance for $\Tarski(n,tk)$.

\begin{definition}[Chain constructor]
Define $\mathcal{C}_t\colon\mathcal{F}^t\rightarrow([n]^{kt}\rightarrow[n]^{kt})$ as follows. Given $\boldsymbol{f}=(f_1,\ldots,f_t)\in\mathcal{F}^t$ and a query point $\boldsymbol{x}=(x_1,\ldots,x_t)\in[n]^{kt}$, with each $x_i\in[n]^k$, the $i$-th $k$-dimensional output block is defined, for every $i\in\{1,\ldots,t\}$, by
\begin{align}
\bigl(\mathcal{C}_t(\boldsymbol{f})(\boldsymbol{x})\bigr)_i :=
\begin{cases}
x_i, & i<\sigma(\boldsymbol{f},\boldsymbol{x}), \\
f_i(x_i), & i=\sigma(\boldsymbol{f},\boldsymbol{x}), \\
T(x_{\sigma(\boldsymbol{f},\boldsymbol{x})}),
& i>\sigma(\boldsymbol{f},\boldsymbol{x}).
\end{cases}
\end{align}
Here, $T$ and $\sigma$ are defined in \eqref{eq:t-padding} and \eqref{eq:sigma-first-unfixed}, respectively.
\end{definition}

We next show that the chain constructor is injective and preserves both the uniqueness of the fixed point and monotonicity. Consequently, $\mathcal{C}_t(\mathcal{F}^t)$ is indeed a hard instance family for $\Tarski(n,tk)$. The antichain assumption on $A$ allows us to apply the following auxiliary lemma.

\begin{lemma} \label{lem:antichain-property}
Let $(\mathcal{L},\preceq)$ be a partially ordered set, and let $A$ be an antichain in $\mathcal{L}$. Consider its downward closure
$D:=\{y\in\mathcal{L} : \exists x\in A,\ y\preceq x\}$.
Then, for any $x\prec y$, where $x\prec y$ means $x\preceq y$ and $x\ne y$, the following statements hold:
\begin{enumerate}
\item If $x\in A$, then $y\in\overline{D}$, where $\overline{D}:=\mathcal{L}\setminus D$ is the complement of $D$ in $\mathcal{L}$.
\item If $y\in A$, then $x\in D\setminus A$.
\end{enumerate}
\end{lemma}

\begin{proof}
For the first statement, suppose for contradiction that $y\in D$. Then there exists $z\in A$ such that $y\preceq z$. By transitivity, $x\prec z$. Since both $x$ and $z$ belong to $A$, this contradicts the antichain property of $A$. Therefore, $y\in\overline{D}$.

For the second statement, the definition of $D$ immediately gives $x\in D$. If $x\in A$ as well, then $x\prec y$ contradicts the antichain property of $A$. Hence, $x\in D\setminus A$.
\end{proof}

We now establish the properties of the chain constructor.

\begin{lemma} \label{lem:chain-constructor-property}
Let $\mathcal{F}$ be a family of hard instances for $\Tarski(n,k)$ such that every $f\in\mathcal{F}$ has a unique fixed point $f^*$ and all possible fixed points lie in an antichain $A$. For a positive integer $t$, let $\mathcal{C}_t$ be the associated chain constructor. Then the following statements hold:
\begin{enumerate}
\item The map $\mathcal{C}_t$ is injective.
\item For every $\boldsymbol{f}=(f_1,\ldots,f_t)\in\mathcal{F}^t$, the function $\mathcal{C}_t(\boldsymbol{f})$ is monotone.
\item For every $\boldsymbol{f}=(f_1,\ldots,f_t)\in\mathcal{F}^t$, the function $\mathcal{C}_t(\boldsymbol{f})$ has a unique fixed point, namely $(f_1^*,\ldots,f_t^*)$.
\end{enumerate}
\end{lemma}

\begin{proof}
For the first statement, consider distinct tuples $\boldsymbol{f}^{(1)},\boldsymbol{f}^{(2)}\in\mathcal{F}^t$. Let $i$ be the smallest index such that $f_i^{(1)}\ne f_i^{(2)}$, and choose a query point $z\in[n]^k$ such that $f_i^{(1)}(z)\ne f_i^{(2)}(z)$. Define $\boldsymbol{x}$ by
\begin{align}
x_j :=
\begin{cases}
(f_j^{(1)})^*, & j<i, \\
z, & j=i, \\
\bot, & j>i.
\end{cases}
\end{align}
For every $j<i$, we have $f_j^{(1)}=f_j^{(2)}$, so $x_j$ is a fixed point of both functions. By the definition of $\mathcal{C}_t$, for each $a\in\{1,2\}$, the $i$-th output block is $\bigl(\mathcal{C}_t(\boldsymbol{f}^{(a)})(\boldsymbol{x})\bigr)_i = f_i^{(a)}(z)$, regardless of whether $z$ is a fixed point of $f_i^{(a)}$. Thus, $\mathcal{C}_t(\boldsymbol{f}^{(1)})(\boldsymbol{x}) \ne \mathcal{C}_t(\boldsymbol{f}^{(2)})(\boldsymbol{x})$, which proves injectivity.

For the second statement, fix $\boldsymbol{f}\in\mathcal{F}^t$ and query points $\boldsymbol{x},\boldsymbol{y}\in[n]^{kt}$ satisfying $\boldsymbol{x}\preceq\boldsymbol{y}$. Let $\ell := \min\{\sigma(\boldsymbol{f},\boldsymbol{x}), \sigma(\boldsymbol{f},\boldsymbol{y})\}$. For every $i<\ell$, both $x_i$ and $y_i$ are the unique fixed point of $f_i$. Hence, $\bigl(\mathcal{C}_t(\boldsymbol{f})(\boldsymbol{x})\bigr)_i=\bigl(\mathcal{C}_t(\boldsymbol{f})(\boldsymbol{y})\bigr)_i=f_i^*$. We consider the following cases:
\begin{enumerate}
\item If $\ell=t+1$, then the definition of $\sigma$ gives $\boldsymbol{x}=\boldsymbol{y}=\mathcal{C}_t(\boldsymbol{f})(\boldsymbol{x})=\mathcal{C}_t(\boldsymbol{f})(\boldsymbol{y})=(f_1^*,\ldots,f_t^*)$. Thus, monotonicity holds.

\item Suppose that $x_\ell$ is a fixed point of $f_\ell$ but $y_\ell$ is not. Then $x_\ell\prec y_\ell$ and $x_\ell\in A$. By Lemma~\ref{lem:antichain-property}, we have $y_\ell\in\overline{D}$. Monotonicity of $f_\ell$ gives $\bigl(\mathcal{C}_t(\boldsymbol{f})(\boldsymbol{x})\bigr)_\ell = f_\ell(x_\ell) \preceq f_\ell(y_\ell) = \bigl(\mathcal{C}_t(\boldsymbol{f})(\boldsymbol{y})\bigr)_\ell$.
For every $i>\ell$, we have $\bigl(\mathcal{C}_t(\boldsymbol{f})(\boldsymbol{x})\bigr)_i \preceq \top = T(y_\ell) = \bigl(\mathcal{C}_t(\boldsymbol{f})(\boldsymbol{y})\bigr)_i$.
Therefore, monotonicity holds.

\item Suppose that $x_\ell$ is not a fixed point of $f_\ell$ but $y_\ell$ is. Then $x_\ell\prec y_\ell$ and $y_\ell\in A$. By Lemma~\ref{lem:antichain-property}, we have $x_\ell\in D\setminus A$. Monotonicity of $f_\ell$ gives $\bigl(\mathcal{C}_t(\boldsymbol{f})(\boldsymbol{x})\bigr)_\ell = f_\ell(x_\ell) \preceq f_\ell(y_\ell) = \bigl(\mathcal{C}_t(\boldsymbol{f})(\boldsymbol{y})\bigr)_\ell$.
For every $i>\ell$, we have $\bigl(\mathcal{C}_t(\boldsymbol{f})(\boldsymbol{x})\bigr)_i = T(x_\ell) = \bot \preceq \bigl(\mathcal{C}_t(\boldsymbol{f})(\boldsymbol{y})\bigr)_i$.
Therefore, monotonicity holds.

\item Suppose that neither $x_\ell$ nor $y_\ell$ is a fixed point of $f_\ell$. Then
$\sigma(\boldsymbol{f},\boldsymbol{x})
=
\sigma(\boldsymbol{f},\boldsymbol{y})
=
\ell$.
The $\ell$-th output blocks are ordered by the monotonicity of $f_\ell$. Moreover, $T$ is monotone because $D$ is downward closed. Thus, for every $i>\ell$, $\bigl(\mathcal{C}_t(\boldsymbol{f})(\boldsymbol{x})\bigr)_i = T(x_\ell) \preceq T(y_\ell) = \bigl(\mathcal{C}_t(\boldsymbol{f})(\boldsymbol{y})\bigr)_i$.
Therefore, monotonicity holds.

\item When $\ell\leq t$, it is impossible for both $x_\ell$ and $y_\ell$ to be fixed points of $f_\ell$, since this would contradict the definition of $\ell$.
\end{enumerate}
These cases establish
$\mathcal{C}_t(\boldsymbol{f})(\boldsymbol{x})
\preceq
\mathcal{C}_t(\boldsymbol{f})(\boldsymbol{y})$.

For the third statement, the definition immediately shows that $(f_1^*,\ldots,f_t^*)$ is a fixed point of $\mathcal{C}_t(\boldsymbol{f})$. Now consider any $\boldsymbol{x}\ne(f_1^*,\ldots,f_t^*)$, and let $i$ be the smallest index such that $x_i\ne f_i^*$. Since $f_i$ has a unique fixed point, $f_i(x_i)\ne x_i$. Therefore, $\sigma(\boldsymbol{f},\boldsymbol{x})=i$, and $\bigl(\mathcal{C}_t(\boldsymbol{f})(\boldsymbol{x})\bigr)_i = f_i(x_i) \ne x_i$.
Hence, $\boldsymbol{x}$ is not a fixed point, proving uniqueness.
\end{proof}

\subsection{Adversary Matrix Construction and Spectral Analysis} \label{subsec:dimensional-spectral-analysis}

In fact, the construction above already yields a bootstrapping of classical deterministic query complexity. We go further by explicitly constructing an adversary matrix to establish a quantum query lower bound.

Let $\mathcal{F}$ be a hard instance family satisfying the assumptions of Theorem~\ref{thm:dimensional-hardness-bootstrapping}, and let $\Gamma_0$ be a nonnegative adversary matrix for $\Tarski(n, k)|_{\mathcal{F}}$ that yields a spectral lower bound of $\lambda / \delta$. Since every entry of $\Gamma_0$ is nonnegative, the Perron--Frobenius theorem guarantees a nonnegative unit eigenvector $v$ satisfying $\Gamma_0 v = \norm{\Gamma_0} v$.

By Lemma~\ref{lem:chain-constructor-property}, $\mathcal{C}_t$ is injective. Thus, each hard instance in $\mathcal{C}_t(\mathcal{F}^t)$ can be uniquely indexed by a tuple $(f_1, \ldots, f_t)\in\mathcal{F}^t$. Given $t$ matrices $M_1, \ldots, M_t$, each of size $|\mathcal{F}|\times|\mathcal{F}|$, we use their tensor product to define a matrix of size $|\mathcal{C}_t(\mathcal{F}^t)|\times|\mathcal{C}_t(\mathcal{F}^t)|$: for any two tuples $\boldsymbol{f}^{(1)}, \boldsymbol{f}^{(2)}\in\mathcal{F}^t$,
\begin{align}
\bigl(M_1\otimes\cdots\otimes M_t\bigr)[\boldsymbol{f}^{(1)}, \boldsymbol{f}^{(2)}]
=
\prod_{i=1}^t M_i[f_i^{(1)}, f_i^{(2)}].
\end{align}

For $i=1,\ldots,t$, define
\begin{align}
\Gamma_i = I^{\otimes(i-1)}\otimes\Gamma_0\otimes J^{\otimes(t-i)},
\end{align}
where $I$ is the $|\mathcal{F}|\times|\mathcal{F}|$ identity matrix and $J=vv^\top$. Define
\begin{align}
    \Gamma = \sum_{i = 1}^t \Gamma_i. \label{eq:chain-adversary-construction}
\end{align}

We show that $\Gamma$ is indeed a nonnegative adversary matrix.

\begin{lemma} \label{lem:chain-adversary-validity}
The matrix $\Gamma$ satisfies the following properties: for any $\boldsymbol{f}^{(1)}, \boldsymbol{f}^{(2)}\in\mathcal{F}^t$,
\begin{enumerate}
\item $\Gamma[\boldsymbol{f}^{(1)}, \boldsymbol{f}^{(2)}] = \Gamma[\boldsymbol{f}^{(2)}, \boldsymbol{f}^{(1)}] \geq 0$;
\item If $\boldsymbol{f}^{(1)*} = \boldsymbol{f}^{(2)*}$, then $\Gamma[\boldsymbol{f}^{(1)}, \boldsymbol{f}^{(2)}] = 0$.
\end{enumerate}
Therefore, $\Gamma$ is a nonnegative adversary matrix for $\Tarski(n, tk)|_{\mathcal{C}_t(\mathcal{F}^t)}$.
\end{lemma}

\begin{proof}
The first property follows immediately from the nonnegativity and symmetry of $I$, $J$, and $\Gamma_0$.

For the second property, $\boldsymbol{f}^{(1)*} = \boldsymbol{f}^{(2)*}$ implies that $(f_i^{(1)})^* = (f_i^{(2)})^*$ for every $i=1,\ldots,t$. Since $\Gamma_0$ is a nonnegative adversary matrix, we have $\Gamma_0[f_i^{(1)}, f_i^{(2)}]=0$. Hence,
\begin{align}
\Gamma_i[\boldsymbol{f}^{(1)}, \boldsymbol{f}^{(2)}]
=
\prod_{j=1}^{i-1} I[f_j^{(1)}, f_j^{(2)}]
\cdot \Gamma_0[f_i^{(1)}, f_i^{(2)}]
\cdot \prod_{j=i+1}^{t} J[f_j^{(1)}, f_j^{(2)}]
=
0.
\end{align}
Since $\Gamma=\sum_{i=1}^t\Gamma_i$, it follows that $\Gamma[\boldsymbol{f}^{(1)}, \boldsymbol{f}^{(2)}]=0$. This establishes the second property.
\end{proof}

We obtain the lower bound on the numerator directly by applying $\Gamma$ to $v^{\otimes t}$:

\begin{lemma} \label{lem:chain-numerator}
$\norm{\Gamma} \geq t\lambda$.
\end{lemma}

\begin{proof}
Since $v$ is a unit Perron vector of $\Gamma_0$, a straightforward calculation gives, for every $i=1,\ldots,t$,
\begin{align}
\Gamma_i v^{\otimes t}
=
v^{\otimes(i-1)} \otimes (\Gamma_0 v) \otimes (vv^\top v)^{\otimes(t-i)}
=
\norm{\Gamma_0} v^{\otimes t}.
\end{align}
Therefore, since $\Gamma=\sum_{i=1}^t\Gamma_i$ and $\norm{v^{\otimes t}}=1$,
\begin{align}
\norm{\Gamma}
\geq
\norm{\Gamma v^{\otimes t}}
=
\norm{t\norm{\Gamma_0}v^{\otimes t}}
=
t\norm{\Gamma_0}
\geq
t\lambda.
\end{align}
\end{proof}

To establish the upper bound on the denominator, we need the following auxiliary lemma. It generalizes \cite[Lemma~2.3]{li2026quantum}, and its proof follows essentially the same argument.

\begin{lemma} \label{lem:almost-orthogonal}
Let $A_1,\ldots,A_m\in\R^{n\times n}$ be real symmetric idempotent matrices, and let $B=\sum_{i=1}^m A_i$. Suppose that there exists a sequence of positive numbers $a_0,a_1,\ldots$ satisfying the following conditions:
\begin{enumerate}
\item The series $\sum_{i=0}^{\infty} a_i$ converges.
\item For every $i,j\in\{1,\ldots,m\}$, we have $\norm{A_iA_j}\leq a_{|i-j|}$.
\end{enumerate}
Then
\begin{align}
\norm{B}\leq a_0+2\sum_{i=1}^{\infty}a_i.
\end{align}
\end{lemma}

\begin{proof}
Consider the linear map $\mathcal{A}\colon\R^n\rightarrow(\R^n)^m$ given by $\mathcal{A}x=(A_1x,\ldots,A_mx)$. Symmetry and idempotence imply $\mathcal{A}^\top\mathcal{A}=B$, so $\norm{B}=\norm{\mathcal{A}\mathcal{A}^\top}$.

The $(i,j)$-block of $\mathcal{A}\mathcal{A}^\top$ is $A_iA_j$. Let $K\in\R^{m\times m}$ have entries $K_{ij}=\norm{A_iA_j}$. Since $(A_iA_j)^\top=A_jA_i$, the matrix $K$ is symmetric. Comparing the norm of a block matrix with the norm of the matrix of its block norms gives
\begin{align}
\norm{B}
=
\norm{\mathcal{A}\mathcal{A}^\top}
\leq
\norm{K}
\leq
\max_{1\leq i\leq m}\sum_{j=1}^m K_{ij}
\leq
a_0+2\sum_{d=1}^{\infty}a_d.
\end{align}
\end{proof}

\begin{lemma} \label{lem:chain-denominator}
For every $\boldsymbol{x}\in[n]^{tk}$, we have $\norm{\Gamma\circ\Delta_{\boldsymbol{x}}}\leq C\delta$, where $C=5+2\sqrt{2}$.
\end{lemma}

\begin{proof}
If $\Gamma_0=0$, the claim is trivial, so assume that $\Gamma_0\ne0$.

We begin by computing $\Gamma_i\circ\Delta_{\boldsymbol{x}}$ for $i\in\{1,\ldots,t\}$. Specifically, for $\boldsymbol{f}^{(1)},\boldsymbol{f}^{(2)}\in\mathcal{F}^t$, we examine the value of $V=(\Gamma_i\circ\Delta_{\boldsymbol{x}})[\boldsymbol{f}^{(1)},\boldsymbol{f}^{(2)}]$ in the following cases.

If there exists $j<i$ such that $f_j^{(1)}\ne f_j^{(2)}$, then $V=0$, since the first $i-1$ tensor factors of $\Gamma_i$ are identity matrices.

Otherwise, suppose that $f_j^{(1)}=f_j^{(2)}$ for every $j<i$. If some such common function does not fix $x_j$, let $j<i$ be the smallest index with this property. Then $\mathcal{C}_t(\boldsymbol{f}^{(1)})(\boldsymbol{x})=\mathcal{C}_t(\boldsymbol{f}^{(2)})(\boldsymbol{x})=(x_1,\ldots,x_{j-1},f_j^{(1)}(x_j),T(x_j),\ldots,T(x_j))$. Hence, $\Delta_{\boldsymbol{x}}[\boldsymbol{f}^{(1)},\boldsymbol{f}^{(2)}]=0$, and thus $V=0$.

Now suppose further that, for every $j<i$, the common function $f_j^{(1)}=f_j^{(2)}$ fixes $x_j$. The following case analysis shows that, on the support of $\Gamma_i$, the chain-query mask agrees with the local-query mask at $x_i$:
\begin{enumerate}
\item Suppose that $x_i$ is a fixed point of exactly one of $f_i^{(1)}$ and $f_i^{(2)}$. Then exactly one of the $i$-th output blocks of $\mathcal{C}_t(\boldsymbol{f}^{(1)})(\boldsymbol{x})$ and $\mathcal{C}_t(\boldsymbol{f}^{(2)})(\boldsymbol{x})$ equals $x_i$, so these blocks differ. Likewise, exactly one of $f_i^{(1)}(x_i)$ and $f_i^{(2)}(x_i)$ equals $x_i$.

\item Suppose that $x_i$ is a fixed point of neither function. All preceding output blocks coincide, since the preceding query points are fixed points. All subsequent padding blocks also coincide, since they depend only on $x_i$. Thus, $\mathcal{C}_t(\boldsymbol{f}^{(1)})(\boldsymbol{x})\ne\mathcal{C}_t(\boldsymbol{f}^{(2)})(\boldsymbol{x})$ if and only if $f_i^{(1)}(x_i)=(\mathcal{C}_t(\boldsymbol{f}^{(1)})(\boldsymbol{x}))_i\ne(\mathcal{C}_t(\boldsymbol{f}^{(2)})(\boldsymbol{x}))_i=f_i^{(2)}(x_i)$.

\item Suppose that $x_i$ is a fixed point of both functions. Then $\Gamma_0[f_i^{(1)},f_i^{(2)}]=0$, and consequently $V=0$.
\end{enumerate}

Combining these cases, we obtain
\begin{align}
(\Gamma_i\circ\Delta_{\boldsymbol{x}})[\boldsymbol{f}^{(1)},\boldsymbol{f}^{(2)}]
=
\mathbf{1}\left\{
\bigwedge_{j=1}^{i-1}
\bigl(f_j^{(1)}=f_j^{(2)}\wedge f_j^{(1)}(x_j)=x_j\bigr)
\;\wedge\;
f_i^{(1)}(x_i)\ne f_i^{(2)}(x_i)
\right\}
\Gamma_i[\boldsymbol{f}^{(1)},\boldsymbol{f}^{(2)}].
\end{align}
Here, $\mathbf{1}\{P\}$ denotes the indicator of a proposition $P$, which equals $1$ if and only if $P$ is true. The identity $\mathbf{1}\{P\wedge Q\}=\mathbf{1}\{P\}\mathbf{1}\{Q\}$ allows us to distribute the conditions in this conjunction among the corresponding tensor factors. For $i=1,\ldots,t$, define
\begin{align}
P_i[f,g]
=
\mathbf{1}\{f=g\wedge f(x_i)=x_i\},
\quad
\widetilde{\Gamma}_i=\Gamma_0\circ\Delta_{x_i},
\end{align}
where $f,g\in\mathcal{F}$. We then obtain
\begin{align}
\Gamma_i\circ\Delta_{\boldsymbol{x}}
=
Q_i
:=
P_1\otimes\cdots\otimes P_{i-1}
\otimes\widetilde{\Gamma}_i
\otimes J^{\otimes(t-i)}.
\end{align}

Let $Q=\sum_{i=1}^t Q_i$. Then $Q=\Gamma\circ\Delta_{\boldsymbol{x}}$, so it suffices to bound the spectral norm of $Q$. To apply Lemma~\ref{lem:almost-orthogonal}, we first associate orthogonal projection surrogates with $Q$ and the $Q_i$. Define
\begin{align}
E_i
=
P_1\otimes\cdots\otimes P_{i-1}
\otimes I
\otimes J^{\otimes(t-i)},
\quad
E=\sum_{i=1}^t E_i.
\end{align}
Since $P_i$, $I$, and $J$ are real symmetric idempotent matrices, each $E_i$ is also real symmetric and idempotent. Moreover, $Q_i=E_i^\top Q_iE_i$ and $\norm{Q_i}\leq\delta$. Taking a unit Perron vector $z$ of $Q$, we have
$$
\begin{aligned}
\norm{Q}
&=
\sum_{i=1}^t z^\top Q_i z
=
\sum_{i=1}^t(E_i z)^\top Q_i(E_i z)\\
&\leq
\delta\sum_{i=1}^t\norm{E_i z}^2
=
\delta z^\top\left(\sum_{i=1}^t E_i\right)z
\leq
\delta\norm{E}.
\end{aligned}
$$

Having established this relation, we turn to the spectral norm of $E$. Assume that $r\leq s$; the other case reduces to this one by taking transposes. If $r=s$, then $\norm{E_r^2}=\norm{E_r}\leq1$. If $r<s$, then
\begin{align}
\norm{E_rE_s}
=
\prod_{j=1}^{r}\norm{P_j}
\cdot
\prod_{j=r+1}^{s-1}\norm{JP_j}
\cdot
\norm{J}^{t-s+1}
\leq
\prod_{j=r+1}^{s-1}\norm{JP_j}.
\end{align}

The quantity $\norm{JP_i}$ has a simple closed form:
\begin{align}
\norm{JP_i}
=
\norm{P_iJ}
=
\max_{w:\norm{w}=1}\norm{P_i vv^\top w}
=
\norm{P_i v}.
\end{align}
Since $P_i\Gamma_0P_i=0$, we have
\begin{align}
\norm{P_i v}
=
\frac{\norm{P_i\Gamma_0v}}{\norm{\Gamma_0}}
=
\frac{\norm{P_i\Gamma_0(I-P_i)v}}{\norm{\Gamma_0}}
\leq
\norm{P_i}\norm{(I-P_i)v}
\leq
\norm{(I-P_i)v}.
\end{align}

The vectors $P_i v$ and $(I-P_i)v$ are orthogonal, and their squared norms sum to $1$. Since $\norm{JP_i}=\norm{P_i v}$ is the smaller of the two norms, it follows that $\norm{JP_i}\leq1/\sqrt{2}$. Hence, $\norm{E_rE_s}\leq(1/\sqrt{2})^{s-r-1}$.

Taking $a_0=1$ and $a_i=(1/\sqrt{2})^{i-1}$ for $i\geq1$, we apply Lemma~\ref{lem:almost-orthogonal} to obtain
\begin{align}
\norm{\Gamma\circ\Delta_{\boldsymbol{x}}}
=
\norm{Q}
\leq
\delta\norm{E}
\leq
\delta\left(1+2\sum_{i=1}^{\infty}a_i\right)
=
C\delta,
\end{align}
where $C=5+2\sqrt{2}$.
\end{proof}

Combining the preceding results, we can now prove Theorem~\ref{thm:dimensional-hardness-bootstrapping}.

\begin{proof}[Proof of Theorem~\ref{thm:dimensional-hardness-bootstrapping}]
Applying the chain constructor $\mathcal{C}_t$ to tuples in $\mathcal{F}^t$ yields the hard instance family $\mathcal{C}_t(\mathcal{F}^t)$. By Lemma~\ref{lem:chain-constructor-property}, this family consists of valid instances of $\Tarski(n, tk)$. Let $\Gamma$ be the matrix defined in \eqref{eq:chain-adversary-construction}. Lemma~\ref{lem:chain-adversary-validity} establishes that $\Gamma$ is a valid nonnegative adversary matrix for this family. Lemmas~\ref{lem:chain-numerator} and~\ref{lem:chain-denominator} provide a lower bound on the numerator and an upper bound on the denominator, respectively. Thus, $\Gamma$ yields a spectral lower bound of $t\lambda/(C\delta)$, where $C=5+2\sqrt{2}$.
\end{proof}

\subsection{Enforcing the Antichain Condition by Conjugation Transformation} \label{subsec:conjugate-transformation}

The antichain condition on the possible fixed points is essential to the preceding construction: it is used to establish the monotonicity of the functions produced by the chain constructor. However, existing hard instance families, including one-dimensional families that are at least as hard as binary search and the two-dimensional herringbone functions from \cite{phillips2026quantum}, do not satisfy this condition. We now give a transformation that lifts the possible fixed points of any hard instance family with unique fixed points to an antichain. 

For any $x\in[n]^k$, define its conjugate by $\overline{x}=(n-1)\mathbf{1}-x$. We define the following antichain constructor.

\begin{definition}[Antichain constructor] \label{def:antichain-constructor}
Define $\mathcal{B}\colon\mathcal{F}\rightarrow([n]^{2k}\rightarrow[n]^{2k})$ as follows: for every $f\in\mathcal{F}$,
\begin{align}
\mathcal{B}(f)(x,y)=(f(x),\overline{f(\overline{y})}).
\end{align}
\end{definition}

The antichain constructor has the following properties.

\begin{lemma}\label{lem:antichain-constructor-property}
Let $\mathcal{F}$ be a hard instance family for $\Tarski(n,k)$ in which every function has a unique fixed point. Then the following statements hold:
\begin{enumerate}
\item Every function in $\mathcal{B}(\mathcal{F})$ is monotone and has a unique fixed point, and the set of all possible fixed points forms an antichain.
\item $\mathcal{B}$ is injective and preserves equality of answers.
\item Let $q_1(x,y)=x$ and $q_2(x,y)=\overline{y}$. For any $f_1,f_2\in\mathcal{F}$ and any query point $(x,y)\in[n]^{2k}$,
\begin{align}
\mathcal{B}(f_1)(x,y)\ne\mathcal{B}(f_2)(x,y)
\quad\Rightarrow\quad
\exists j\in\{1,2\},\quad
f_1(q_j(x,y))\ne f_2(q_j(x,y)).
\end{align}
\end{enumerate}
\end{lemma}

\begin{proof}
For the first statement, we first establish monotonicity. Consider any $(x_1,y_1)\preceq(x_2,y_2)$, so that $x_1\preceq x_2$ and $y_1\preceq y_2$. For the first output block, $x_1\preceq x_2$ implies $f(x_1)\preceq f(x_2)$. For the second output block, $y_1\preceq y_2$ successively implies $\overline{y_2}\preceq\overline{y_1}$, $f(\overline{y_2})\preceq f(\overline{y_1})$, and $\overline{f(\overline{y_1})}\preceq\overline{f(\overline{y_2})}$. Thus, $\mathcal{B}(f)$ is monotone. Moreover, its unique fixed point is $(f^*,\overline{f^*})$. All possible fixed points belong to the antichain $\{z\in[n]^{2k}:\sum_{i=1}^{2k}z_i=k(n-1)\}$, and hence themselves form an antichain.

The second statement follows immediately from the definition and the form of the fixed points.

The third statement follows directly from the definition.
\end{proof}

The preceding lemma shows that the antichain constructor is an injective map satisfying the hypotheses of Lemma~\ref{lem:adversary-reduction}. Consequently, we obtain the following theorem.

\antichainlifting*

\begin{proof}
This follows immediately from Lemmas~\ref{lem:antichain-constructor-property} and~\ref{lem:adversary-reduction}.
\end{proof}

\section{Logarithmic Hardness Bootstrapping Theorem} \label{sec:logarithmic-bootstrapping}

Intuitively, solving $\Tarski$ amounts to performing binary search with $\Tarski^*$ as a subroutine. In certain settings, however, $\Tarski^*$ is no easier than $\Tarski$. We therefore use it as an intermediate problem and apply Theorem \ref{thm:generalized-search-composition} to prove the following theorem.

\logarithmicbootstrapping*

The proof proceeds in three steps. In Section~\ref{subsec:slice-transform}, we construct from $\mathcal{F}$ a $\Tarski^*$ problem in which an arbitrary sign $\sigma$ can be hidden at the fixed point. In Section~\ref{subsec:composition}, we show that this problem is a generalized search problem, allowing us to compose it with ordered search using Theorem \ref{thm:generalized-search-composition}. In Section~\ref{subsec:embedding}, we embed the resulting abstract problem into a staircase function on a grid of width $nm$ and dimension $2k$.

The antichain condition is essential to this argument: it allows the sign hidden at the fixed point to be chosen arbitrarily. However, the embedding in the third step does not preserve this condition. To accumulate logarithmic factors through repeated bootstrapping, we must therefore apply the antichain constructor (Definition~\ref{def:antichain-constructor}) after the embedding. Thus, gaining an additional $\log m$ factor increases the dimension from $k$ to $4k$ and the grid width from $n$ to $nm$. The resulting dimension $4k$ is smaller than the dimensions $4k+2$ and $4k+1$ in \cite{chen2026quasifpt} and \cite{liu2026dimension}, respectively.

\subsection{\texorpdfstring
  {A Hidden-Sign Transformation from {\mdseries\textsc{Tarski}} to {\mdseries\textsc{Tarski}}$^*$}
  {A Hidden-Sign Transformation from Tarski to Tarski*}} \label{subsec:slice-transform}

Recall that an instance of $\Tarski^*(n,k)$ consists of a
monotone function $f\colon[n]^k\to[n]^k$ and a monotone
sign function $\sigma\colon[n]^k\to\{-1,0,1\}$.
Since every fixed point of $f$ is a solution to
$\Tarski^*$, whenever the solution is unique, it must
coincide with the unique fixed point $f^*$ of $f$.
Restricting to such instances, we define the
\emph{hidden-sign problem} as determining the sign
at this fixed point:
\begin{align}
\Tarski^H(n,k)\colon(f,\sigma)\mapsto\sigma(f^*).
\end{align}

In this subsection, we prove the following lemma.

\begin{lemma}
Let $\mathcal{F}$ be a hard instance family for $\Tarski(n,k)$ in which every function has a unique fixed point and all possible fixed points lie on an antichain $A$. Suppose that there exists a nonnegative adversary matrix $\Gamma_0$ for $\Tarski(n,k)|_{\mathcal{F}}$ that yields a spectral lower bound of $\lambda/\delta$. Then there exists a hard instance family $\mathcal{F}'$ for $\Tarski^H(n,k)$, together with a corresponding nonnegative adversary matrix $\Gamma$ that yields a spectral lower bound of $\lambda/\delta$.
\end{lemma}

To this end, we define the following hidden-sign constructor, which hides a sign $s\in\{-1,0,1\}$.

\begin{definition}[Hidden-sign constructor]\label{def:slice-constructor}
Let $\mathcal{F}$ be a hard instance family for $\Tarski(n,k)$ in which every function has a unique fixed point and all possible fixed points lie on an antichain $A$. For each $s\in\{-1,0,1\}$, define $\mathcal{H}_s\colon\mathcal{F}\rightarrow(([n]^k\rightarrow[n]^k)\times([n]^k\rightarrow\{-1,0,1\}))$ as follows: for every $f\in\mathcal{F}$,
\begin{align}
\mathcal{H}_s(f)=(f,\sigma_s),\quad
\text{where }\sigma_s(x)=
\begin{cases}
s, & x=f^*,\\
-1, & x\ne f^*\wedge x\in D,\\
1, & x\in\overline{D}.
\end{cases}
\end{align}
Here, $D=\{y\in[n]^k:\exists x\in A,\ y\preceq x\}$ is the downward closure of $A$.
\end{definition}

Regardless of the sign to be hidden, the following lemma ensures that $\mathcal{H}_s(\mathcal{F})$ consists of valid instances of $\Tarski^H$.

\begin{lemma}
Let $\mathcal{F}$ be a hard instance family for $\Tarski(n,k)$ in which every function has a unique fixed point and all possible fixed points lie on an antichain $A$. Then, for every $s\in\{-1,0,1\}$, the following statements hold:
\begin{enumerate}
\item $\mathcal{H}_s$ is injective.
\item For every $f\in\mathcal{F}$, the sign function $\sigma_s$ is monotone.
\item For every $f\in\mathcal{F}$, $f^*$ is the unique solution to the $\Tarski^*$ instance specified by $(f,\sigma_s)$.
\end{enumerate}
\end{lemma}

\begin{proof}
The first statement is immediate by considering the first component.

For the second statement, away from $f^*$, $\sigma_s$ equals $-1$ on $D$ and $1$ on $\overline{D}$, which is consistent with monotonicity because $D$ is downward closed. By Lemma~\ref{lem:antichain-property}, every point strictly below $f^*$ lies in $D\setminus A$, whereas every point strictly above $f^*$ lies in $\overline{D}$. Thus, assigning any value $s\in\{-1,0,1\}$ at $f^*$ preserves monotonicity. Moreover, the other points of $A$ are incomparable with $f^*$, so their values impose no additional constraints on this assignment.

For the third statement, clearly $f^*$ is a solution to the $\Tarski^*$ instance for every choice of $s$. Now consider any $x\ne f^*$. By uniqueness, $x\ne f(x)$. If $x$ and $f(x)$ are incomparable, then $x$ satisfies neither of the two conditions defining a $\Tarski^*$ solution. Otherwise, we exclude the following two cases:
\begin{enumerate}
\item Suppose that $x\prec f(x)$. Since $f$ is monotone, it maps the interval $\{y\in[n]^k:y\succeq x\}$ into itself. By Tarski's fixed-point theorem, $f$ has a fixed point in this interval, which must be $f^*$ by uniqueness. Hence, $x\prec f^*$. Lemma~\ref{lem:antichain-property} gives $x\in D\setminus A$, so $\sigma_s(x)=-1$. Therefore, $x$ is not a solution to the $\Tarski^*$ instance.
\item Suppose that $x\succ f(x)$. Since $f$ is monotone, it maps the interval $\{y\in[n]^k:y\preceq x\}$ into itself. By Tarski's fixed-point theorem, $f$ has a fixed point in this interval, which must be $f^*$ by uniqueness. Hence, $x\succ f^*$. Lemma~\ref{lem:antichain-property} gives $x\in\overline{D}$, so $\sigma_s(x)=1$. Therefore, $x$ is not a solution to the $\Tarski^*$ instance.
\end{enumerate}
Thus, $f^*$ is the unique solution.
\end{proof}

Let 
\begin{align}
\mathcal{F}'=\mathcal{H}_{-1}(\mathcal{F})\sqcup\mathcal{H}_{0}(\mathcal{F})\sqcup\mathcal{H}_{1}(\mathcal{F}). \label{eq:tarski-h-construction}
\end{align} 
The preceding lemma shows that $\mathcal{F}'$ consists of valid instances of $\Tarski^H$. Every element of $\mathcal{F}'$ can be uniquely labeled by a pair $(s,f)$, where $s\in\{-1,0,1\}$ and $f\in\mathcal{F}$. We therefore construct the following adversary matrix. Let $K_3$ be the adjacency matrix of a triangle:
\begin{align}
\Gamma=K_3\otimes\Gamma_0,\quad
\text{where }K_3=
\begin{pmatrix}
0 & 1 & 1\\
1 & 0 & 1\\
1 & 1 & 0
\end{pmatrix}.
\end{align}

We show that $\Gamma$ scales both the numerator and the denominator in the adversary ratio by a factor of $\norm{K_3}=2$, and hence still yields a spectral lower bound of $\lambda/\delta$.

\begin{lemma} \label{lem:tarski-h-lowerbound}
$\Gamma$ is a nonnegative adversary matrix for $\Tarski^H(n, k)|_{\mathcal{F}'}$ and yields a spectral lower bound of $\lambda/\delta$.
\end{lemma}

\begin{proof}
The symmetry and nonnegativity of $\Gamma$ follow immediately from those of $\Gamma_0$ and $K_3$. Since the diagonal entries of $K_3$ are zero, $s_1=s_2$ implies $\Gamma[(s_1,f_1),(s_2,f_2)]=0$. Thus, $\Gamma$ is a nonnegative adversary matrix.

For the lower bound on the numerator, it suffices to observe that
\begin{align}
\norm{\Gamma}
=\norm{K_3\otimes\Gamma_0}
=\norm{K_3}\cdot\norm{\Gamma_0}
\geq 2\lambda.
\end{align}

For the upper bound on the denominator, let $\Delta'_x$ and $\Delta_x$ denote the distinguisher matrices at $x$ for the new and original problems, respectively. We show that $\Gamma\circ\Delta'_x=K_3\otimes(\Gamma_0\circ\Delta_x)$ for every $x\in[n]^k$.

Consider two inputs $(s_1,f_1)$ and $(s_2,f_2)$. When $s_1\ne s_2$, if $x$ is a common fixed point of $f_1$ and $f_2$, then $\mathcal{H}_{s_1}(f_1)(x)\ne\mathcal{H}_{s_2}(f_2)(x)$, but $\Gamma_0[f_1,f_2]=0$. Otherwise, $\mathcal{H}_{s_1}(f_1)(x)\ne\mathcal{H}_{s_2}(f_2)(x)$ if and only if $f_1(x)\ne f_2(x)$. On the other hand, when $s_1=s_2$, we have $K_3[s_1,s_2]=0$. These observations establish the claimed matrix identity. Therefore,
\begin{align}
\norm{\Gamma\circ\Delta'_x}
=\norm{K_3\otimes(\Gamma_0\circ\Delta_x)}
=2\norm{\Gamma_0\circ\Delta_x}
\leq 2\delta.
\end{align}

Consequently, $\Gamma$ yields a spectral lower bound of $(2\lambda)/(2\delta)=\lambda/\delta$, proving the lemma.
\end{proof}

\subsection{\texorpdfstring
  {Composing {\mdseries\textsc{Tarski}}$^*$ with Ordered Search}
  {Composing Tarski* with Ordered Search}} \label{subsec:composition}

We next compose $\Tarski^H$ with $\OS_m$. We first need to show that $\Tarski^H|_{\mathcal{F}'}$ is a generalized search problem. This follows directly from the definition of $\mathcal{F}'$.

\begin{lemma}
$\Tarski^H(n, k)|_{\mathcal{F}'}$ is a generalized search problem.
\end{lemma}

\begin{proof}
By the definition of $\mathcal{F}'$, each of the output symbols $-1,0,1$ has exactly $|\mathcal{F}|$ preimages. For any query point $x$ and two inputs labeled by $(s_1,f_1),(s_2,f_2)\in\mathcal{F}'$, we distinguish the following cases according to whether $x$ is a fixed point of $f_1$ and $f_2$. This criterion depends only on $x,f_1,f_2$.
\begin{enumerate}
\item If $x$ is a fixed point of neither $f_1$ nor $f_2$, then $\sigma_{s_1}(x)=\sigma_{s_2}(x)$. Thus, the oracle answers differ if and only if $f_1(x)\ne f_2(x)$, which, for this fixed query, depends only on $f_1,f_2$.
\item If $x$ is a fixed point of exactly one of $f_1$ and $f_2$, then $f_1(x)\ne f_2(x)$. Hence, the oracle answers always differ, regardless of $s_1,s_2$.
\item If $x$ is a fixed point of both $f_1$ and $f_2$, then $f_1(x)=f_2(x)=x$, and $\sigma_{s_1}(x)\ne\sigma_{s_2}(x)$ if and only if $s_1\ne s_2$. Thus, whether the oracle answers differ depends only on whether $s_1=s_2$, not on their specific values.
\end{enumerate}
Therefore, for each query point $x$, equality of the oracle answers depends only on $f_1,f_2$ and whether $s_1=s_2$. It follows that $\Tarski^H(n, k)|_{\mathcal{F}'}$ is a generalized search problem.
\end{proof}

This satisfies the hypotheses of Theorem~\ref{thm:generalized-search-composition}. Consequently, we obtain the following lemma.

\begin{lemma} \label{lem:abstract-composition}
There exists an absolute constant $C>0$ such that, for every integer $m\geq 2$, the composition of $\OS_m$ with $m$ copies of $\Tarski^H(n,k)|_{\mathcal{F}'}$ admits a nonnegative adversary matrix $\Gamma$ witnessing the lower bound
\begin{align}
\SA\bigl(\OS_m\circ(
\Tarski^H(n,k)|_{\mathcal{F}'},\ldots,
\Tarski^H(n,k)|_{\mathcal{F}'}
)\bigr)
\geq \SA(\OS_m) \cdot
\SA\bigl(\Tarski^H(n,k)|_{\mathcal{F}'}\bigr).
\end{align}
\end{lemma}

\begin{proof}
Applying Theorem~\ref{thm:generalized-search-composition} yields the claimed inequality.
\end{proof}

We denote the domain of the composit problem by $\mathcal{D}$, and label an element in $\mathcal{D}$ by $(s_1, f_1), \ldots, (s_m, f_m)$.

\subsection{Embedding of the Composed Problem} \label{subsec:embedding}

Finally, we embed the hard instances constructed above into a grid lattice. By Definition~\ref{def:composite-function}, query positions of the composed problem can be labeled by pairs $(i,x)\in\{1,\ldots,m\}\times[n]^k$. We therefore consider the following set:
\begin{align}
Q=\{1,\ldots,m\}\times[n]^k
=\bigsqcup_{i=1}^m P_i,\quad
\text{where }P_i=\{(i,x):x\in[n]^k\}.
\end{align}

Equip $Q$ with the partial order $\preceq$ defined as follows: for any $(i,x),(j,y)\in Q$, we have $(i,x)\preceq(j,y)$ if and only if $i<j$, or $i=j$ and $x\preceq y$. Although $Q$ is an abstract lattice, it admits an order-preserving embedding into the grid lattice $[nm]^k$. Consider the maps $e\colon Q\rightarrow[nm]^k$ and $e^\dagger\colon[nm]^k\rightarrow Q$ defined by
\begin{equation}
\begin{aligned}
e(i,x)&=n(i-1)\mathbf{1}+x,\\
a(z)&=\max_{j\in\{1,\ldots,k\}}
\left\lfloor\frac{z_j}{n}\right\rfloor,\\
e^\dagger(z)&=
\bigl(a(z)+1,\,
\max(z,na(z)\mathbf{1})-na(z)\mathbf{1}\bigr).
\end{aligned}
\end{equation}
Here, the $\max$ in the third line is taken entrywise. Both maps are clearly order-preserving, and $e^\dagger(e(i,x))=(i,x)$ for every $(i,x)\in Q$.

\subsubsection{Reducing the Embedding to a Staircase}

It turns out that directly embedding the composed problem into $Q$ is difficult. We therefore consider a staircase $S$ in $Q^2$, equipped with the product order: for any $(x_1,y_1),(x_2,y_2)\in Q^2$, we have $(x_1,y_1)\preceq(x_2,y_2)$ if and only if $x_1\preceq x_2\wedge y_1\preceq y_2$. The lattice $Q^2$ admits an order-preserving embedding into $[nm]^{2k}$. Define
\begin{align}
S=\bigcup_{i=1}^m D_i
\cup\bigcup_{i=1}^{m-1}B_i,
\end{align}
where $D_i=P_i\times P_i$ is called a \emph{main block}, and $B_i=P_i\times P_{i+1}$ is called a \emph{bridge}. The following retraction $\rho\colon Q^2\rightarrow S$ maps points outside the staircase onto it and fixes every point of the staircase:
\begin{align}
\rho((a,x),(b,y))=
\begin{cases}
((a,x),(a,\bot)), & b<a,\\
((a,x),(b,y)), & b=a\vee b=a+1,\\
((b-1,\bot),(b,y)), & b>a+1.
\end{cases}
\end{align}
The monotonicity of $\rho$ follows from a routine case analysis, which we omit. We illustrate $e, e^\dagger$ and $\rho$ in Figure~\ref{fig:embedding-maps}.

\begin{figure}[htbp]
    \centering
    \includegraphics[width=\linewidth]{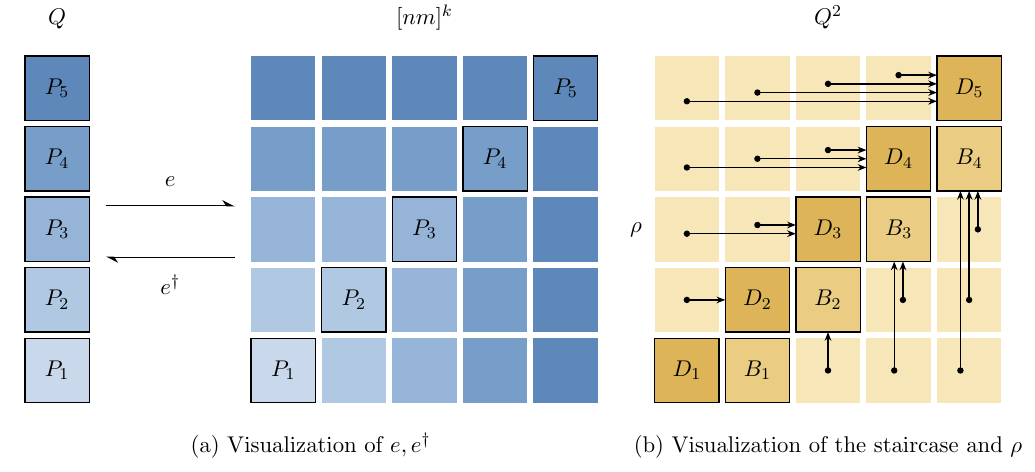}
    \caption{Visualization of the embeddings: $e, e^\dagger, \rho$ and the staircase.}
    \label{fig:embedding-maps}
\end{figure}

Our task thus reduces to constructing a staircase constructor $\hat{\mathcal{E}}\colon\mathcal{D}\rightarrow(S\rightarrow S)$ that produces a monotone self-map of $S$ for each input. The order-preserving maps and retractions above then ensure that the \emph{embedding constructor}
\begin{align}
\mathcal{E}\colon(\boldsymbol{s},\boldsymbol{f})
\mapsto
(e\times e)\circ
\hat{\mathcal{E}}(\boldsymbol{s},\boldsymbol{f})
\circ\rho\circ
(e^\dagger\times e^\dagger)
\end{align}
transforms each instance of the abstract composed problem into an instance of $\Tarski(nm,2k)$.

\subsubsection{The Staircase Hard Instance Constructor}

We now construct the staircase constructor $\hat{\mathcal{E}}$ described above. This construction also explains our choice of the staircase $S\subseteq Q^2$ as its domain.

Consider a hard instance $(\boldsymbol{s},\boldsymbol{f})\in\mathcal{D}$ of the abstract composed problem. Here, the ordered search promise gives $s_i=+1$ before the hidden position, $s_i=0$ at the hidden position, and $s_i=-1$ after it. We use the sign functions of the inner $\Tarski^*$ instances to determine whether a point is mapped to the preceding block, the following block, or its current block. The staircase domain ensures that, whenever a point moves to an adjacent block, one of its two components $(i_t,x_t)$, with $t\in\{1,2\}$, retains its block index $i_t$, while the other block index changes from $i_{3-t}$ to $i_{3-t}\pm1$. We map the component whose block index remains unchanged to $(i_t,f_{i_t}(x_t))$, thereby revealing information about the corresponding inner function and using its monotonicity. The component that changes blocks is mapped to $(i_{3-t}+1,\bot)$ when moving forward, or to $(i_{3-t}-1,\top)$ when moving backward.

\begin{remark}
The extremal padding values $\top$ and $\bot$ are used to preserve monotonicity across blocks. Indeed, the order on $Q$ is lexicographic: whenever $i_1<i_2$, we have $(i_1,x_1)\preceq(i_2,x_2)$, regardless of $x_1,x_2$. Thus, comparability of inputs in different blocks imposes no relation on their local coordinates. Padding the coordinate part of a transferred component with the appropriate extremal element ensures that the required inequalities still hold when the output block indices coincide.
\end{remark}

Specifically, we define the transition directions on $D_i$ and $B_i$ by
\begin{align}
\varphi_i(x,y)
&=
\begin{cases}
+1, & \sigma_{s_i}(x)=+1,\\
-1, & \sigma_{s_i}(x)<+1\wedge\sigma_{s_i}(y)=-1,\\
0, & \text{otherwise},
\end{cases}\\
\psi_i(x,y)
&=
\begin{cases}
-1, & \sigma_{s_i}(x)<+1,\\
0, & \sigma_{s_i}(x)=+1\wedge\sigma_{s_{i+1}}(y)=-1,\\
+1, & \sigma_{s_i}(x)=+1\wedge\sigma_{s_{i+1}}(y)>-1.
\end{cases}
\end{align}
These transition rules are chosen to ensure that: (1) The transition directions are accurate when $x, y$ are both fixed points; (2) both $\varphi_i$ and $\psi_i$ are monotone; (3) comparable points in adjacent blocks cannot cross their shared boundary in opposite directions; and (4) under the ordered search promise on $\boldsymbol{s}$, a point formed by the fixed points of $f_i$ and $f_{i+1}$ cannot remain in $B_i$. We now define the staircase constructor.

\begin{definition}[Staircase constructor]\label{def:staircase-constructor}
Given $(\boldsymbol{s},\boldsymbol{f})\in\mathcal{D}$, write $F=\hat{\mathcal{E}}(\boldsymbol{s},\boldsymbol{f})$. For a point $((i,x),(i,y))\in D_i$, let $s=\varphi_i(x,y)$ and define
\begin{align}
F((i,x),(i,y))
=
\begin{cases}
((i-1,\top),(i,f_i(y))), & s=-1\wedge i>1,\\
((i,f_i(x)),(i+1,\bot)), & s=+1\wedge i<m,\\
((i,f_i(x)),(i,f_i(y))), & \text{otherwise}.
\end{cases}
\end{align}

For a point $((i,x),(i+1,y))\in B_i$, let $s=\psi_i(x,y)$ and define
\begin{align}
F((i,x),(i+1,y))
=
\begin{cases}
((i,f_i(x)),(i,\top)), & s=-1,\\
((i,f_i(x)),(i+1,f_{i+1}(y))), & s=0,\\
((i+1,\bot),(i+1,f_{i+1}(y))), & s=+1.
\end{cases}
\end{align}
\end{definition}

The following lemma establishes the properties of $\hat{\mathcal{E}}$.

\begin{lemma} \label{lem:staircase-constructor-property}
The following statements hold for the staircase constructor $\hat{\mathcal{E}}$:
\begin{enumerate}
\item Both $\hat{\mathcal{E}}$ and $\mathcal{E}$ are injective.
\item For every $(\boldsymbol{s},\boldsymbol{f})\in\mathcal{D}$, the function $\hat{\mathcal{E}}(\boldsymbol{s},\boldsymbol{f})$ is monotone.
\item For every $(\boldsymbol{s},\boldsymbol{f})\in\mathcal{D}$, the function $\hat{\mathcal{E}}(\boldsymbol{s},\boldsymbol{f})$ has a unique fixed point, namely $((i, f_i^*), (i, f_i^*))$ where $s_i = 0$.
\item Consider any two instances $(\boldsymbol{s}^{(1)},\boldsymbol{f}^{(1)}),(\boldsymbol{s}^{(2)},\boldsymbol{f}^{(2)})\in\mathcal{D}$ and any query point $u=((i,x),(j,y))\in S$, where $j=i$ or $j=i+1$. Define $q_1(u)=(i,x)$ and $q_2(u)=(j,y)$. For $a\in\{1,2\}$, let $h^{(a)}$ denote the oracle function specified by $(\boldsymbol{s}^{(a)},\boldsymbol{f}^{(a)})$, so that $h^{(a)}(r,z)=\mathcal{H}_{s_r^{(a)}}(f_r^{(a)})(z)$. Then
\begin{align}
\hat{\mathcal{E}}(\boldsymbol{s}^{(1)},\boldsymbol{f}^{(1)})(u)
\ne
\hat{\mathcal{E}}(\boldsymbol{s}^{(2)},\boldsymbol{f}^{(2)})(u)
\quad\Rightarrow\quad
\exists \ell\in\{1,2\},\quad
h^{(1)}(q_\ell(u))\ne h^{(2)}(q_\ell(u)).
\end{align}
\end{enumerate}
\end{lemma}

\begin{proof}
For the first statement, we show that the input can be recovered from the resulting staircase function. For any $i$ and $z\in[n]^k$, query the diagonal point $((i,z),(i,z))\in D_i$. Regardless of the transition direction, at least one output component remains in block $i$, and its local coordinate is $f_i(z)$. Thus, all functions $f_i$ can be recovered. Next, query $((i,f_i^*),(i,f_i^*))$. At this point, $\varphi_i(f_i^*,f_i^*)=s_i$. The ordered search promise ensures that $s_1\ne-1$ and $s_m\ne+1$, so the boundary conditions do not obscure the transition direction. The output block indices therefore determine $s_i$. Hence, $\hat{\mathcal{E}}$ is injective. Since $\hat{\mathcal{E}}$ can be recovered from $\mathcal{E}$ by restricting to $(e\times e)(S)$ and applying $e^\dagger\times e^\dagger$, $\mathcal{E}$ is also injective.

For the second statement, the monotonicity of the sign functions $\sigma_{s_i}$ immediately implies that of $\varphi_i$ and $\psi_i$. Fix $(\boldsymbol{s},\boldsymbol{f})\in\mathcal{D}$ and write $F=\hat{\mathcal{E}}(\boldsymbol{s},\boldsymbol{f})$. Consider two points $u=((i_1,x_1),(j_1,y_1))\preceq v=((i_2,x_2),(j_2,y_2))$, where $j_a=i_a$ or $j_a=i_a+1$ for $a\in\{1,2\}$. We distinguish the following cases:
\begin{enumerate}
\item Both points lie in the same block. Thus, $i_1=i_2=i$, $j_1=j_2=j$, $x_1\preceq x_2$, and $y_1\preceq y_2$. Monotonicity follows from that of $f_i,f_j$ and the relevant transition function $\varphi_i$ or $\psi_i$, together with the extremal padding values.

\item The points lie in adjacent blocks. The only potentially problematic case is that of opposing cross-block transitions. In all other cases, monotonicity follows from the corresponding componentwise comparisons. Figure~\ref{fig:adjacent-block-monotonicity} illustrates some representative cases; the remaining cases require only straightforward case analysis, which we omit.
\begin{enumerate}
\item The first point lies in $D_i$ and the second in $B_i$. Here, $i_1=i_2=i$, $j_1=i$, $j_2=i+1$, and $x_1\preceq x_2$. The only potentially incompatible case is when the two points cross their shared boundary in opposite directions. By the definitions of $\varphi_i$ and $\psi_i$, this requires $\sigma_{s_i}(x_1)=+1$ and $\sigma_{s_i}(x_2)<+1$, contradicting the monotonicity of $\sigma_{s_i}$. All remaining cases are ordered by the monotonicity of $f_i$ and the extremal padding values.

\item The first point lies in $B_i$ and the second in $D_{i+1}$. Here, $i_1=i$, $i_2=i+1$, $j_1=j_2=i+1$, and $y_1\preceq y_2$. Again, the only potentially incompatible case is when the two points cross their shared boundary in opposite directions. By the definitions of $\varphi_{i+1}$ and $\psi_i$, this requires $\sigma_{s_{i+1}}(y_1)>-1$ and $\sigma_{s_{i+1}}(y_2)=-1$, contradicting the monotonicity of $\sigma_{s_{i+1}}$. All remaining cases are ordered by the monotonicity of $f_{i+1}$ and the extremal padding values.
\end{enumerate}

\item The points lie in nonadjacent blocks. We induct on the distance between their blocks in the sequence $D_1,B_1,D_2,\ldots,B_{m-1},D_m$. If $u\in D_i$, then
\begin{align}
u
\preceq ((i,x_1),(i+1,\bot))
\preceq ((i+1,\bot),(i+1,\bot))
\preceq v.
\end{align}
The first two comparisons involve adjacent blocks, while the last comparison has a smaller block distance. Thus, the preceding cases and the induction hypothesis give $F(u)\preceq F(v)$. If $u\in B_i$, take the intermediate point $w=((i+1,\bot),(i+1,y_1))\in D_{i+1}$. Then $u\preceq w\preceq v$, and the same inductive argument applies.
\end{enumerate}
This establishes monotonicity. 

\begin{figure}[t]
    \centering
    \includegraphics[width=\linewidth]{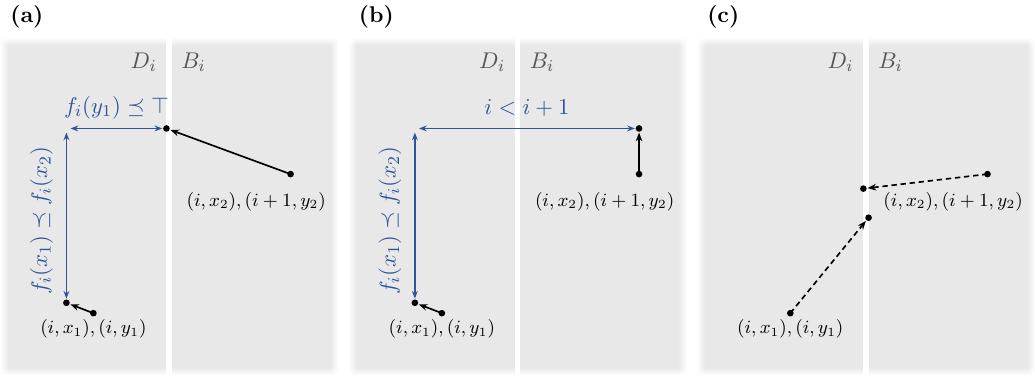}
    \caption{Three representative cases in the monotonicity analysis for adjacent blocks. Consider $u=((i,x_1),(i,y_1))\in D_i$ and $v=((i,x_2),(i+1,y_2))\in B_i$, where $x_1\preceq x_2$ and no order relation between $y_1$ and $y_2$ is assumed. (a) When $u$ remains in $D_i$ and $v$ is mapped from $B_i$ into $D_i$, the order is preserved. (b) When both points remain in their respective blocks, the order is also preserved. (c) Opposing cross-block transitions are the only potentially nonmonotone configuration, but this configuration cannot occur.}
    \label{fig:adjacent-block-monotonicity}
\end{figure}

For the third statement, let $i^*$ be the unique index such that $s_{i^*}=0$. We show that the unique fixed point is $((i^*,f_{i^*}^*),(i^*,f_{i^*}^*))$.

First, no fixed point can lie in a bridge. A fixed point in $B_i$ would require $x=f_i^*$, $y=f_{i+1}^*$, and $\psi_i(x,y)=0$. By the definition of $\psi_i$, this would imply $\sigma_{s_i}(f_i^*)=s_i=+1$ and $\sigma_{s_{i+1}}(f_{i+1}^*)=s_{i+1}=-1$, contradicting the ordered search promise.

Next, any fixed point in a main block $D_i$ must satisfy $x=y=f_i^*$. Suppose that $s_i\ne0$. A fixed point could then arise only in one of the following three cases:
\begin{enumerate}
\item $i=1$ and $x=y=f_1^*$. Since $s_1\ne0$, the ordered search promise gives $s_1=+1$. Consequently, $((1,f_1^*),(1,f_1^*))$ is mapped into $B_1$ and cannot be a fixed point.
\item $i=m$ and $x=y=f_m^*$. Since $s_m\ne0$, the ordered search promise gives $s_m=-1$. Consequently, $((m,f_m^*),(m,f_m^*))$ is mapped into $B_{m-1}$ and cannot be a fixed point.
\item $1<i<m$, $x=y=f_i^*$, and $\varphi_i(x,y)=0$. The last condition implies $\sigma_{s_i}(x)=s_i<+1$ and $\sigma_{s_i}(y)=s_i>-1$. Hence, $s_i=0$, contradicting the assumption.
\end{enumerate}
Therefore, a fixed point can lie only in $D_{i^*}$, where it must be $((i^*,f_{i^*}^*),(i^*,f_{i^*}^*))$. This point is indeed fixed, since $\varphi_{i^*}(f_{i^*}^*,f_{i^*}^*)=0$. Thus, the fixed point is unique.

The fourth statement follows directly from the definition of the staircase constructor: its value at $u$ is determined entirely by the two oracle answers at $q_1(u)$ and $q_2(u)$.
\end{proof}

\subsubsection{Putting Everything Together}

The properties established in the preceding lemma, together with those of $e^\dagger$, $\rho$, and $e$, immediately carry over to $\mathcal{E}$. Consequently, $\mathcal{E}(\mathcal{D})$ is a valid hard instance family for $\Tarski(nm,2k)$.

\begin{proof}[Proof of Theorem~\ref{thm:logarithm-bootstrapping}]
Starting from $\mathcal{F}$ and a spectral lower bound of $\lambda/\delta$, we proceed as follows:
\begin{enumerate}
\item Apply the construction in \eqref{eq:tarski-h-construction} to obtain a hard instance family $\mathcal{F}'$ for which $\Tarski^H(n,k)|_{\mathcal{F}'}$ is a generalized search problem. By Lemma~\ref{lem:tarski-h-lowerbound}, there exists a nonnegative adversary matrix that yields a spectral lower bound of $\lambda/\delta$.
\item Apply Lemma~\ref{lem:abstract-composition} to obtain the domain $\mathcal{D}$ of the abstract composed problem \begin{align}
\OS_m\circ(\Tarski^H(n,k)|_{\mathcal{F}'},\ldots,\Tarski^H(n,k)|_{\mathcal{F}'}),
\end{align} 
with $m$ copies of the inner problem, together with a nonnegative adversary matrix that yields a spectral lower bound of $\SA(\OS_m)\lambda/\delta$.
\item Apply the embedding constructor $\mathcal{E}$ to obtain the hard instance family $\mathcal{E}(\mathcal{D})$ for $\Tarski(nm,2k)$. By Lemma~\ref{lem:staircase-constructor-property}, each query to an embedded instance depends on only two queries to the original instance. Moreover, its fixed point determines the ordered search answer. Lemma~\ref{lem:adversary-reduction} therefore gives a nonnegative adversary matrix that yields a spectral lower bound of $\SA(\OS_m)\lambda/(2\delta)$.
\end{enumerate}
\end{proof}

\section{Consequences of Hardness Bootstrapping} \label{sec:consequences}

Finally, we apply our two hardness bootstrapping theorems to derive the following corollaries, which answer several previously open questions. All our lower bounds hold in the quantum query model and therefore also apply to both deterministic and randomized classical algorithms.

\subsection{The Bootstrapping Scheme}

In general, we apply the bootstrapping theorems according to the following scheme. Starting from a unique-fixed-point hard instance family for the two-dimensional problem $\Tarski(n_0,2)$, where $n_0\geq2$, we choose an integer $m\geq2$ and a nonnegative integer $r$ specifying the parameter and the number of applications of logarithmic bootstrapping, respectively. We work with spectral adversary bounds throughout the construction.
\begin{enumerate}
\item Apply the antichain constructor, doubling the dimension and dividing the lower bound by $2$.
\item Repeat the following steps $r$ times:
\begin{enumerate}
\item Apply the logarithmic hardness bootstrapping theorem, multiplying the width by $m$, the dimension by $2$, and the lower bound by $\SA(\OS_m)/2$.
\item Apply the antichain constructor, doubling the dimension and dividing the lower bound by $2$, thereby restoring the antichain condition.
\end{enumerate}
\item The current dimension is $d_r=4^{r+1}$. For a target dimension $k$, set $t=\lfloor k/d_r\rfloor$ and apply the dimensional hardness bootstrapping theorem, multiplying the lower bound by $t/(5+2\sqrt{2})$. If necessary, pad the remaining dimensions to reach $k$.
\end{enumerate}

The resulting width is $n_0m^r$. Thus, for target parameters $n\geq n_0$ and $k\geq4$, the parameters must satisfy
\begin{align}
0\leq r\leq
\min\left\{
\left\lfloor\log_4(k/4)\right\rfloor,\,
\left\lfloor\log_m(n/n_0)\right\rfloor
\right\}.
\label{eq:bootstrapping-parameter-budget}
\end{align}
Any remaining width can also be padded without loss in the adversary bound.

If the initial spectral lower bound is $L_0$, then the bound before dimensional bootstrapping is $(L_0/2)(\SA(\OS_m)/4)^r$. Since $k\geq d_r$, we have $t=\lfloor k/d_r\rfloor\geq k/(2d_r)$. Thus, extracting a linear factor of $k$ incurs an additional factor of $4^{-r}$. Combining this scheme with the $\Omega(\log^2 n_0)$ spectral adversary lower bound for the two-dimensional problem, we obtain, for every fixed $\epsilon\in[0,1/2)$, a constant $R>0$ such that
\begin{align}
Q_\epsilon(\Tarski(n,k))
\geq
Rk\log^2 n_0
\left(\frac{\SA(\OS_m)}{16}\right)^r.
\label{eq:bootstrapping-lowerbound}
\end{align}
Here, $R$ absorbs the bounded-error factor $(1-2\sqrt{\epsilon(1-\epsilon)})/2$, the loss from rounding $t$, and the other fixed constants. In particular, it is independent of $n,k,n_0,m,r$.

We first examine several special cases of this scheme.

For constant dimensions, we use a variant that omits both the final application of the antichain constructor and the dimensional bootstrapping theorem. For $r\geq1$, the resulting dimension is $2\cdot4^r$. Fix a constant dimension $k\geq8$. Set $r=\lfloor\log_4(k/2)\rfloor$ and $n_0=m=\lfloor n^{1/(r+1)}\rfloor$. For sufficiently large $n$, these choices satisfy $2\cdot4^r\leq k$ and $n_0m^r\leq n$, and yield a quantum query lower bound of $\Omega_{k,\epsilon}(\log^{r+2} n)$. For $2\leq k<8$, we instead pad the original two-dimensional family without applying either bootstrapping theorem. Together, these bounds immediately yield:

\begin{corollary}
For every fixed integer $k\geq2$ and every fixed $\epsilon\in[0,1/2)$, we have $Q_\epsilon(\Tarski(n,k))=\log^{\Omega(\log k)} n$.
\end{corollary}

For $k\geq2$, we have $Q_\epsilon(\Tarski(n,k))=\Omega(\log^2 n)$; for $k\geq8$, we have $Q_\epsilon(\Tarski(n,k))=\Omega(\log^3 n)$; and so on: each fourfold increase in the dimension yields an additional $\log n$ factor in the lower bound. This rules out an upper bound of the form $f(k)\log^2 n$. Moreover, together with the upper bounds in \cite{chen2026quasifpt}, it shows that for $k=8,9,10$, the query complexity is tight up to a constant factor, namely $\Theta(\log^3 n)$.

Setting $r=0$ and $n_0=n$ immediately yields the following lower bound, capturing the joint dependence on $k$ and $n$:

\begin{corollary}[A Dimension-Linear Lower Bound $\Omega(k\log^2 n)$]
For every fixed $\epsilon\in[0,1/2)$ and all integers $k\geq4$ and $n\geq2$,
\begin{align}
Q_\epsilon(\Tarski(n,k))=\Omega(k\log^2 n).
\end{align}
\end{corollary}

The earlier classical dimension-sensitive lower bounds of \cite{branzei2025tarski,branzei2025randomized} together give $\Omega(k+k\log^2 n/\log k)$, while \cite{li2026quantum} established the dimension-sensitive quantum lower bound $\Omega(k\log n)$. Our bound improves or matches these earlier bounds throughout the parameter range $n\geq2$ and $k\geq4$.

\subsection{Polynomial-in-Width Lower Bounds in High Dimensions}
\label{subsec:polynomial-width-tradeoff}

Whenever the dimension is
at least a fixed positive power of the width, our bootstrapping theorems yield
a lower bound that is linear in the dimension and polynomial in the width.

Fix an integer $b\geq203$, and set $n_0=2$, $m=2^b$, and
$r=\lfloor\log_m(n/2)\rfloor-1$. For sufficiently large $n$, we have $r\geq0$,
and the bootstrapping scheme produces a hard instance family of width
$2m^r\leq n$ and dimension
\begin{align}
d_r=4^{r+1}\leq(n/2)^{2/b}.
\label{eq:tradeoff-dimension-budget}
\end{align}
Each round of logarithmic bootstrapping followed by antichain lifting
multiplies the spectral lower bound by $\SA(\OS_m)/4$ and the dimension by $4$.
Thus, after applying dimensional bootstrapping to reach dimension $k\geq d_r$,
the resulting quantum query lower bound is
\begin{align}
Q_\epsilon(\Tarski(n,k))
\geq Rk\left(\frac{\SA(\OS_m)}{16}\right)^r,
\label{eq:tradeoff-bootstrap-bound}
\end{align}
where $R>0$ depends only on the fixed error parameter
$\epsilon\in[0,1/2)$.

By Lemma~\ref{lem:ordered-search-lowerbound}, define
\begin{align}
A_b:=\frac{b\ln2-2}{16\pi}>1,
\qquad
\delta_b:=\frac{\log A_b}{b}.
\label{eq:tradeoff-exponent}
\end{align}
Since $r\geq\log_m(n/2)-2$, we obtain
\begin{align}
Q_\epsilon(\Tarski(n,k))
\geq RkA_b^r
\geq RA_b^{-2}k(n/2)^{\delta_b}
=\Omega_{b,\epsilon}(kn^{\delta_b}).
\end{align}

\begin{corollary}
\label{cor:polynomial-width-tradeoff}
For every fixed $\alpha>0$, let
\begin{align}
b(\alpha):=\max\left\{203,\left\lceil\frac{2}{\alpha}\right\rceil\right\},
\qquad
\delta(\alpha):=
\frac{1}{b(\alpha)}
\log\left(\frac{b(\alpha)\ln2-2}{16\pi}\right).
\label{eq:dimension-exponent-tradeoff}
\end{align}
For every fixed $\epsilon\in[0,1/2)$ and all sufficiently large $n$, if
$k\geq n^\alpha$, then
\begin{align}
Q_\epsilon(\Tarski(n,k))
=\Omega_{\alpha,\epsilon}(kn^{\delta(\alpha)}).
\end{align}
Moreover, $\delta(\alpha)>0$ for every $\alpha>0$.
\end{corollary}

\begin{proof}
The choice of $b=b(\alpha)$ ensures that $2/b\leq\alpha$. Hence
\eqref{eq:tradeoff-dimension-budget} gives $d_r\leq n^\alpha\leq k$,
and the preceding lower bound applies.
\end{proof}

\begin{figure}[htbp]
\centering
\includegraphics[width=0.7\textwidth]{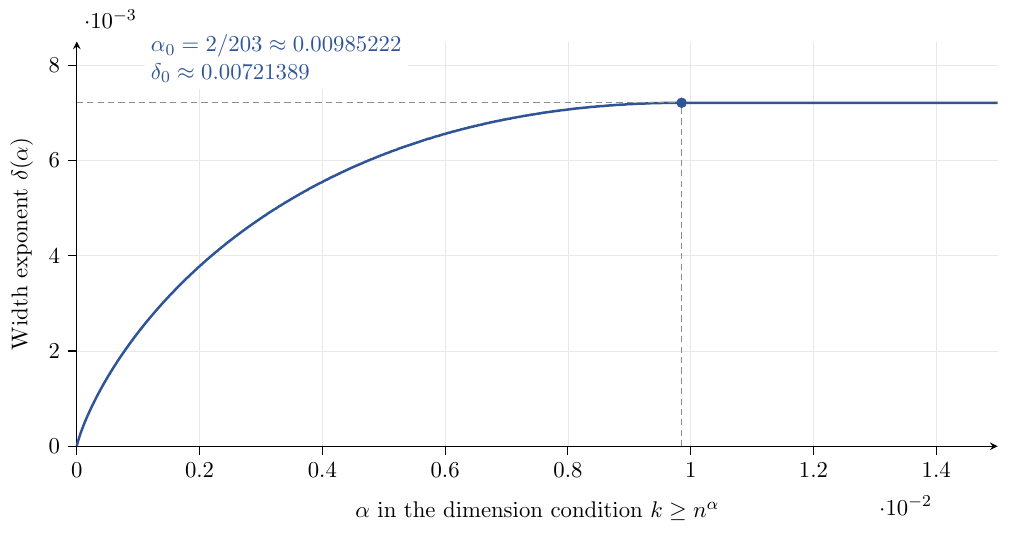}
\caption{The trade-off between the dimension threshold $k\geq n^\alpha$
and the width exponent in the lower bound
$\Omega(kn^{\delta(\alpha)})$.}
\label{fig:polynomial-width-tradeoff}
\end{figure}

The trade-off in
Figure~\ref{fig:polynomial-width-tradeoff} reaches its largest exponent at
\begin{align}
\alpha_0=\frac{2}{203}\approx0.00985221675,
\qquad
\delta_0=\frac{1}{203}
\log\left(\frac{203\ln2-2}{16\pi}\right)
\approx0.00721389254,
\end{align}
and remains at this value for all $\alpha\geq\alpha_0$.
In particular, for $k\geq n^{2/203}$, we obtain
\begin{align}
Q_\epsilon(\Tarski(n,k))=\Omega(kn^{0.00721389}).
\end{align}
This also implies the simpler bound $\Omega(kn^{0.0072})$.
More generally, the trade-off establishes simultaneous linear dependence on
$k$ and polynomial dependence on $n$ throughout every regime $k\geq n^\alpha$
with fixed $\alpha>0$.

\subsection{Superpolynomial Lower Bounds at Exponential Width}

In this subsection, we consider exponentially wide grids, namely $n=2^{\Theta(k)}$. It suffices to analyze $n=2^{\ell k}$ for a fixed constant $\ell>0$, with integer rounding understood. For sufficiently large $k$, choose the following parameters:
\begin{enumerate}
\item $r=\lfloor\log_4(k/4)\rfloor$;
\item $b=\lfloor\ell k/(r+2)\rfloor$, with $n_0=2^{2b}$ and $m=2^b$.
\end{enumerate}

These parameters satisfy the required constraints:
\begin{align}
4\cdot4^r\leq k,\qquad
n_0m^r=2^{b(r+2)}\leq2^{\ell k}=n.
\end{align}

Since $\log n_0=2b$ and $\log m=b$, the bootstrapping scheme gives
\begin{align}
Q_\epsilon(\Tarski(2^{\ell k},k))
\geq
4Rkb^2\left(\frac{\SA(\OS_{2^b})}{16}\right)^r.
\end{align}
For fixed $\ell>0$, we have $r=\Theta(\log k)$ and $b=\Theta(k/\log k)$. Moreover, Lemma~\ref{lem:ordered-search-lowerbound} gives $\SA(\OS_{2^b})\geq(b\ln2-2)/\pi=\Omega(b)$, with an absolute implied constant. Consequently, we obtain the following corollary.

\begin{corollary}
For every fixed $\ell>0$ and $\epsilon\in[0,1/2)$,
\begin{align}
Q_\epsilon(\Tarski(2^{\ell k},k))
=
\left(\frac{k}{\log k}\right)^{\Omega(\log k)}.
\end{align}
\end{corollary}

This is a superpolynomial lower bound. Reductions from problems such as $\SSG$ and $\ARRIVAL$ to $\Tarski$ have motivated the hope that efficient Tarski algorithms could help establish $\SSG,\ARRIVAL\in\textsf{P}$ (or $\BQP$, in the quantum regime). The corollary rules out an approach based on invoking a generic polynomial-query solver for arbitrary black-box Tarski instances in this parameter regime.

\section{Conclusions} \label{sec:conclusion}

In summary, we prove two hardness bootstrapping theorems for quantum query lower bounds of finding Tarski fixed points. Suppose that $\mathcal{F}$ is a hard instance family for $\Tarski(n, k)$ such that every function in $\mathcal{F}$ has a unique fixed point, all possible fixed points lie on an antichain $A$, and a nonnegative adversary matrix yields a spectral lower bound of $L$. Then:
\begin{enumerate}
\item Applying the chain constructor yields a spectral lower bound of $tL/(5+2\sqrt{2})$ for $\Tarski(n, tk)$.
\item Composing the associated hidden-sign problem with ordered search $\OS_m$ and embedding the composed problem into $[nm]^{2k}$ yields a spectral lower bound of $CL\log m/2$ for $\Tarski(nm, 2k)$, where $C>0$ is an absolute constant.
\end{enumerate}

By applying these two bootstrapping theorems in sequence, we obtain several significant consequences, including: (1) An $\Omega(\log^{\log_4(k / 2)} n)$ lower bound for constant dimensions; (2) An $\Omega(k\log^2 n)$ lower bound for $k\geq 4$; (3) For every fixed constant $\alpha>0$, there exists a constant $\delta(\alpha)>0$ such that $Q_\epsilon(\Tarski(n,k))=\Omega(kn^{\delta(\alpha)})$ whenever $k>n^\alpha$. (4) On exponentially wide grids with $n=2^{\Theta(k)}$, the quantum query lower bound is superpolynomial in $k$. 

Finally, we highlight several limitations of our approach that suggest directions for future research. The polylogarithmic exponent in the classical query upper bound of \cite{chen2026quasifpt} increases by one when the dimension increases from $k$ to $2k+1$. In contrast, iterating our lower-bound bootstrapping requires increasing the dimension from $k$ to $4k$ to gain an additional factor of $\log m$. The antichain property is a crucial invariant in this bootstrapping process. However, our current antichain constructor doubles the dimension, while logarithmic hardness bootstrapping does not preserve the antichain property. To close this gap with the classical upper bound, two natural questions are: \emph{For the hard instance families produced by logarithmic bootstrapping, can an antichain constructor be implemented in dimension $k+1$? Can logarithmic hardness bootstrapping be achieved in dimension $2k+2$ while preserving the antichain property?} The primitive hard instance family in \cite{li2026quantum}, whose possible fixed points already form an antichain, may provide useful ideas for addressing these questions.

\section*{Acknowledgments}
We thank Yuhao Li for useful discussions on their upper bound in his paper \cite{chen2026quasifpt}. We used GPT-6-Pro to assist with developing constructions, performing algebraic calculations, and polishing the manuscript. The final manuscript was substantially revised by the authors, who remain solely responsible for all mathematical claims, proofs, references, and conclusions. Tongyang Li and Ziyi Yang are supported by the National Natural Science Foundation of China under Grant No.~62372006.

\bibliographystyle{alphaurl}
\bibliography{references}

\end{document}